\documentclass{SCIYA2017enOL}
\newtheorem{notation}[theorem]{Notation}

\usepackage{ulem}
\usepackage{multicol}
\online
\begin{document}

\ensubject{fdsfd}
\ArticleType{ARTICLES}
\Year{2023}
\Month{}%
\Vol{to appear}
\No{}
\BeginPage{1} %
\DOI{}
\ReceiveDate{January 5, 2023}
\AcceptDate{April 17, 2023}
\OnlineDate{January 1, 2022}

\title[]{Parametric ``Non-nested" Discriminants\\for Multiplicities of Univariate Polynomials\footnotemark[4]\footnotetext[4]{This paper has been accepted for publication in SCIENCE CHINA Mathematics.}}{Parametric ``Non-nested" Discriminants for Multiplicities of Univariate Polynomials}

\author[1]{Hong Hong}{hong@ncsu.edu}
\author[2,$\ast$]{Jing Yang}{{yangjing0930@gmail.com}}

\AuthorMark{Hoon Hong}

\AuthorCitation{Hoon Hong, Jing Yang}

\address[1]{Department of Mathematics, North Carolina State University, Box 8205, Raleigh, NC {\rm 27695}, USA}
\address[2]{SMS--HCIC--School of Mathematics and Physics, Center for Applied Mathematics of Guangxi,\\Guangxi Minzu University, Nanning {\rm 530006}, China}

\abstract{We consider the problem of complex root classification, i.e., finding the
conditions on the coefficients of a univariate polynomial for all possible
multiplicity structures on its complex roots. It is well known that such
conditions can be written as conjunctions of several polynomial equations and
one inequation in the coefficients. Those polynomials in the coefficients are
called discriminants for multiplicities. It is also known that discriminants
can be obtained by using repeated parametric gcd's. The resulting
discriminants are usually nested determinants, that is, determinants of
matrices whose entries are determinants, and so on. In this paper, we give a
new type of discriminants which are not based on repeated gcd's. The new
discriminants are simpler in the sense that they are non-nested  determinants
and have smaller maximum degrees.
}



\maketitle

\section{Introduction}

In this paper, we consider the problem of complex root classification, i.e.,
finding the conditions on the coefficients of a polynomial {over the complex
field $\mathbb{C}$} for every potential multiplicity structure its complex
roots may have. For example, consider a quintic polynomial $F=a_{5}x^{5}%
+a_{4}x^{4}+a_{3}x^{3}+a_{2}x^{2}+a_{1}x+a_{0}$ {where $a_{i}$'s take values
over~$\mathbb{C}$}. We would like to find conditions $C_{0},C_{1},\ldots
,C_{6}$ on $a=(a_{0},\ldots,a_{5})$ such that%
\[
\text{multiplicity structure of }F =\left\{
\begin{array}
[c]{lll}%
\left(  1,1,1,1,1\right),  & \text{if } & C_{0}\left(  a\right)  \ \text{holds;}%
\\
(2,1,1,1), & \text{if} & C_{1}\left(  a\right)  \ \text{holds};\\
(2,2,1), & \text{if} & C_{2}\left(  a\right)  \ \text{holds};\\
(3,1,1), & \text{if} & C_{3}\left(  a\right)  \ \text{holds};\\
(3,2), & \text{if} & C_{4}\left(  a\right)  \ \text{holds};\\
\left(  4,1\right),  & \text{if} & C_{5}\left(  a\right)  \ \text{holds};\\
\left(  5\right),  & \text{if} & C_{6}\left(  a\right)  \ \text{holds}.%
\end{array}
\right.
\]

\noindent In general, the problem is stated as follows:

\noindent\textbf{Problem}: \emph{For every} $\boldsymbol{\mu}=\left(  \mu
_{1},\ldots,\mu_{m}\right)  $ \emph{such that} $\mu_{1}\ge\ldots\ge\mu_{m}>0$
\emph{and} $\mu_{1}+\cdots+\mu_{m}=n$\emph{, find a condition on the
coefficients of a polynomial }$F$\emph{{over $\mathbb{C}$} of degree~}$n$ such
that the multiplicity structure of $F$ is $\boldsymbol{\mu}$.

The problem is important because many tasks in mathematics, science
and engineering can be reduced to the problem. Due to its importance, the
problem and several related problems have been already carefully studied
\cite{1998_Gonzalez_Recio_Lombardi,2021_Hong_Yang,2006_Liang_Jeffrey,2008_Liang_Jeffrey_Maza,1999_Liang_Zhang,1996_Yang_Hou_Zeng}%
.

The problem can be viewed as a generalization of the well known problem of
finding a condition on coefficients such that the polynomial has the given
number of distinct roots. This subproblem has been extensively studied. For
instance, the subdiscriminant theory provides a complete solution to the
subproblem: a univariate polynomial of degree $n$ has $m$ distinct roots if
and only if its $0$-th, $\ldots$, $(n-m-1)$-th psd's (i.e., principal
subdiscriminant coefficient) vanish and the $(n-m)$-th psd does not. For
details, see standard textbooks on computational
algebra~(e.g.,~\cite{2006_Basu_Pollack_Roy}).

In \cite{1996_Yang_Hou_Zeng}, Yang, Hou and Zeng gave an algorithm to generate
conditions for discriminating different multiplicity structures of a
univariate polynomial (referred as YHZ's condition hereinafter) by making use
of repeated gcd computation for parametric polynomials
\cite{1971_Brown_Traub,1967_Collins,1983_Loos}. It is based on a similar idea
adopted by Gonzalez-Vega et al. \cite{1998_Gonzalez_Recio_Lombardi} for
solving the real root classification and quantifier elimination problems by
using Sturm-Habicht sequences. The conditions produced by these methods are
conjunctions of several polynomial equations and one inequation on the
coefficients. Those polynomials in the coefficients are called discriminants
for multiplicities. The maximum degree of the discriminants grows
exponentially in the degree of $F$. Furthermore, each discriminant is a
``nested'' determinant, that is, it is a determinant of a matrix whose entries
are again determinants and so on.

In \cite{2021_Hong_Yang}, the authors developed a new type of multiplicity
discriminants to distinguish different multiplicities when the number of
distinct roots is fixed. The main idea is to convert the multiplicity
condition expressed as a permanent inequation in roots into a sum of
determinants in coefficients. In order to generate conditions for all the
possible multiplicity structures of a univariate polynomial, one may first use
subdiscriminants in classical resultant theory to decide the number of
distinct complex roots and then add one more inequation to discriminate
different multiplicity structures with the same number of distinct roots. In
the new condition, the maximum degree of the discriminants grows linearly in
the degree of $F$, which makes the size of discriminants significantly
smaller. However, the form of resulting discriminants is a sum of many
determinants, which makes the further analysis (reasoning) difficult.

The main contribution in this paper is to provide a new type of discriminants,
which are \emph{non-nested} determinants and whose maximum degrees are
\emph{smaller} than those in the previous methods. The method is based on a
significantly different theory and techniques from the previous methods (which
are essentially based on repeated parametric gcd or subdiscriminant theory).
The new condition is given by a newly devised multiplicity discriminant in
coefficients for every potential multiplicity vector of a given degree, which
can be viewed as a generalization of subdiscriminant theory to higher order
derivatives. To build up the connection between the new discriminants and
multiple roots, we first convert it into the ratio of two determinants in
terms of generic roots (without considering the multiplicities). Then by
making use of the connection between divided difference with multiple nodes
and the derivatives of higher orders at the nodes, we integrate the
multiplicity information into the expression and convert it into an expression
in terms of multiple roots. After careful manipulation, it is shown that the
new discriminant can capture the multiplicity information.

The paper is structured as follows. In Section \ref{sec:problem}, we first
present the problem to be solved in a formal way. In Section \ref{sec:main},
we give a precise statement of the main result of the paper (Theorem
\ref{thm:main_result}). Then a proof of Theorem \ref{thm:main_result} is
provided in Section \ref{sec:proof}. The proof is long thus we divide the
proof into three subsections which are interesting on their own. In Section
\ref{sec:comparison}, we compare the form and size of polynomials in the
multiplicity-discriminant condition in Theorem \ref{thm:main_result} and those
given by previous works.

\section{Problem}

\label{sec:problem}

\begin{definition}
[Multiplicity vector]Let $F\in\mathbb{C}\left[  x\right]  $ with $m$ distinct
complex roots, say $r_{1},\ldots,r_{m}$, with multiplicities $\mu_{1}%
,\ldots,\mu_{m}$ respectively. Without losing generality, we assume that
$\mu_{1}\geq\cdots\geq\mu_{m}>0$. Then the \emph{multiplicity} vector of~$F$,
written as $\operatorname*{mult}\left(  F\right)  $, is defined by%
\[
\operatorname*{mult}\left(  F\right)  =\left(  \mu_{1},\ldots,\mu_{m}\right).
\]

\end{definition}

\begin{example}
Let $F=x^{5}-5x^{4}+7x^{3}+x^{2}-8x+4$. Then $\operatorname*{mult}\left(
F\right)  =\left(  2,2,1\right)  $, since it can be verified that $F$
$=\left(  x-1\right)  ^{2}\left(  x+1\right)  ^{1}\left(  x-2\right)  ^{2}$.
Note that the multiplicity vector is a partition of $5,$ which is the degree
of $F$.
\end{example}

\begin{definition}
[Potential multiplicity vectors]Let $n$ be a positive integer. Let
$\mathcal{M}(n)$ stand for the set of all the potential multiplicity vectors
of polynomials of degree $n$, equivalently, the set of all partitions of $n,$
that is,
\[
\mathcal{M}(n)=\left\{  (\mu_{1},\ldots,\mu_{m}):\,\mu_{1}+\cdots+\mu
_{m}=n,\mu_{1}\geq\cdots\geq\mu_{m}>0\right\}.
\]

\end{definition}

\begin{example}
$\mathcal{M}\left(  5\right)  =\left\{  \ \left(  1,1,1,1,1\right)
,\ \ \left(  2,1,1,1\right)  ,\ \left(  2,2,1\right)  ,\ \left(  3,1,1\right)
,\ \left(  3,2\right)  ,\ \left(  4,1\right)  ,\ \left(  5\right)  \ \right\}
$.
\end{example}


\begin{problem}
[Parametric multiplicity problem]The parametric multiplicity problem is stated as:

\begin{enumerate}
\item[In\ :] $n$, a positive integer standing for the polynomial of degree $n
$ with parametric coefficients $a$, that is,%
\[
F=\sum_{i=0}^{n}a_{i}x^{i}\ \ \text{where}\ \ a_{n}\neq0.
\]

\item[Out:] For each $\boldsymbol{\mu}\in\mathcal{M}(n)$, find a condition
$C_{\boldsymbol{\mu}}\ $on $a\ $such that $\operatorname*{mult}\left(
F\right)  =\boldsymbol{\mu}$.
\end{enumerate}
\end{problem}

\section{Main Result}

\label{sec:main}

\begin{definition}
[Determinant polynomial]Consider a vector of univariate polynomials
\[
P=\left[
\begin{array}
[c]{c}%
P_{0}\\
\vdots\\
P_{k}%
\end{array}
\right]  \in\mathbb{C}[x]^{k+1}%
\]
where $\deg P_{i}\leq k\ \ \text{and}\ \ P_{i}=\sum_{0\leq j\leq k}a_{ij}%
x^{j}. $ The \emph{coefficient matrix }of $P,$ written as $C\left(  P\right)
,$ is defined by
\[
C\left(  P\right)  =\operatorname*{coef}\left(  P\right)  =%
\begin{bmatrix}
\operatorname*{coef}\left(  P_{0}\right) \\
\vdots\\
\operatorname*{coef}\left(  P_{k}\right)
\end{bmatrix}
=%
\begin{bmatrix}
a_{0k} & \cdots & a_{00}\\
\vdots &  & \vdots\\
a_{kk} & \cdots & a_{k0}%
\end{bmatrix}.
\]
The \emph{determinant polynomial} of $P,$ written as $\operatorname*{dp}%
\left(  P\right)  ,$ is defined by%
\[
\operatorname*{dp}\left(  P\right)  =| C\left(  P\right) |.
\]
\end{definition}

\begin{definition}
[Multiplicity Discriminant]Let $F=\sum_{i=0}^{n}a_{i}x^{i}\ \ $where$\ \ a_{n}%
\neq0$. Let $\boldsymbol{\gamma}=\left(  \gamma_{1},\ldots,\gamma_{s}\right)
\in\mathcal{M}\left(  n\right)  $. The the $\boldsymbol{\gamma}$%
-\emph{discriminant }of $F$, written as $D\left(  \boldsymbol{\gamma}\right)
,\ $is defined by%
\[
D\left(  \boldsymbol{\gamma}\right)  =\frac{1}{a_{n}}\operatorname*{dp}\left[
\begin{array}
[c]{l}%
F^{(0)}x^{\gamma_{0}-1}\\
~~~~\ \vdots\\
F^{(0)}x^{0}\\\hline
F^{(1)}x^{\gamma_{1}-1}\\
~~~~\ \vdots\\
F^{(1)}x^{0}\\\hline
~~~~\ \vdots\\\hline
F^{(s)}x^{\gamma_{s}-1}\\
~~~~\ \vdots\\
F^{(s)}x^{0}%
\end{array}
\right]
\]
where $\gamma_{0}$ is the smallest so that the above matrix is square
{and $F^{(i)}$ is the $i$-th derivative of $F$ in terms of $x$}. It is straightforward to show that $\gamma_{0}=\gamma_{1}-1$.
\end{definition}

\begin{example}
Let $n=5$ and $F=\sum_{i=0}^{n}a_{i}x^{i}$ and $a_{n}\neq0$. Then%
\[%
\begin{array}
[c]{lll}%
D\left(  5\right)  & = & \operatorname*{dp}\left[
\begin{array}
[c]{l}%
F^{(0)}x^{3}\\
F^{(0)}x^{2}\\
F^{(0)}x^{1}\\
F^{(0)}x^{0}\\\hline
F^{(1)}x^{4}\\
F^{(1)}x^{3}\\
F^{(1)}x^{2}\\
F^{(1)}x^{1}\\
F^{(1)}x^{0}%
\end{array}
\right]  =\dfrac{1}{a_{5}}\left\vert
\begin{array}
[c]{ccccccccc}%
a_{5} & a_{4} & a_{3} & a_{2} & a_{1} & a_{0} &  &  & \\
& a_{5} & a_{4} & a_{3} & a_{2} & a_{1} & a_{0} &  & \\
&  & a_{5} & a_{4} & a_{3} & a_{2} & a_{1} & a_{0} & \\
&  &  & a_{5} & a_{4} & a_{3} & a_{2} & a_{1} & a_{0}\\\hline
5a_{5} & 4a_{4} & 3a_{3} & 2a_{2} & 1a_{1} &  &  &  & \\
& 5a_{5} & 4a_{4} & 3a_{3} & 2a_{2} & 1a_{1} &  &  & \\
&  & 5a_{5} & 4a_{4} & 3a_{3} & 2a_{2} & 1a_{1} &  & \\
&  &  & 5a_{5} & 4a_{4} & 3a_{3} & 2a_{2} & 1a_{1} & \\
&  &  &  & 5a_{5} & 4a_{4} & 3a_{3} & 2a_{2} & 1a_{1}%
\end{array}
\right\vert,\\
D\left(  4,1\right)  & = & \operatorname*{dp}\left[
\begin{array}
[c]{l}%
F^{(0)}x^{2}\\
F^{\left(  0\right)  }x^{1}\\
F^{(0)}x^{0}\\\hline
F^{(1)}x^{3}\\
F^{(1)}x^{2}\\
F^{(1)}x^{1}\\
F^{(1)}x^{0}\\\hline
F^{(2)}x^{0}%
\end{array}
\right]  =\dfrac{1}{a_{5}}\left\vert
\begin{array}
[c]{cccccccc}%
a_{5} & a_{4} & a_{3} & a_{2} & a_{1} & a_{0} &  & \\
& a_{5} & a_{4} & a_{3} & a_{2} & a_{1} & a_{0} & \\
&  & a_{5} & a_{4} & a_{3} & a_{2} & a_{1} & a_{0}\\\hline
5a_{5} & 4a_{4} & 3a_{3} & 2a_{2} & 1a_{1} &  &  & \\
& 5a_{5} & 4a_{4} & 3a_{3} & 2a_{2} & 1a_{1} &  & \\
&  & 5a_{5} & 4a_{4} & 3a_{3} & 2a_{2} & 1a_{1} & \\
&  &  & 5a_{5} & 4a_{4} & 3a_{3} & 2a_{2} & 1a_{1}\\\hline
&  &  &  & 5\cdot4a_{5} & 4\cdot3a_{4} & 3\cdot2a_{3} & 2\cdot1a_{2}%
\end{array}
\right\vert,\\
D\left(  3,2\right)  & = & \operatorname*{dp}\left[
\begin{array}
[c]{l}%
F^{(0)}x^{1}\\
F^{(0)}x^{0}\\\hline
F^{(1)}x^{2}\\
F^{(1)}x^{1}\\
F^{(1)}x^{0}\\\hline
F^{(2)}x^{1}\\
F^{(2)}x^{0}%
\end{array}
\right]  =\dfrac{1}{a_{5}}\left\vert
\begin{array}
[c]{ccccccc}%
a_{5} & a_{4} & a_{3} & a_{2} & a_{1} & a_{0} & \\
& a_{5} & a_{4} & a_{3} & a_{2} & a_{1} & a_{0}\\\hline
5a_{5} & 4a_{4} & 3a_{3} & 2a_{2} & 1a_{1} &  & \\
& 5a_{5} & 4a_{4} & 3a_{3} & 2a_{2} & 1a_{1} & \\
&  & 5a_{5} & 4a_{4} & 3a_{3} & 2a_{2} & 1a_{1}\\\hline
&  & 5\cdot4a_{5} & 4\cdot3a_{4} & 3\cdot2a_{3} & 2\cdot1a_{2} & \\
&  &  & 5\cdot4a_{5} & 4\cdot3a_{4} & 3\cdot2a_{3} & 2\cdot1a_{2}%
\end{array}
\right\vert,
\end{array}
\]
\[%
\begin{array}
[c]{lll}%
D\left(  3,1,1\right)  & = & \operatorname*{dp}\left[
\begin{array}
[c]{l}%
F^{(0)}x^{1}\\
F^{(0)}x^{0}\\\hline
F^{(1)}x^{2}\\
F^{(1)}x^{1}\\
F^{(1)}x^{0}\\\hline
F^{(2)}x^{0}\\\hline
F^{(3)}x^{0}%
\end{array}
\right]  =\dfrac{1}{a_{5}}\left\vert
\begin{array}
[c]{ccccccc}%
a_{5} & a_{4} & a_{3} & a_{2} & a_{1} & a_{0} & \\
& a_{5} & a_{4} & a_{3} & a_{2} & a_{1} & a_{0}\\\hline
5a_{5} & 4a_{4} & 3a_{3} & 2a_{2} & 1a_{1} &  & \\
& 5a_{5} & 4a_{4} & 3a_{3} & 2a_{2} & 1a_{1} & \\
&  & 5a_{5} & 4a_{4} & 3a_{3} & 2a_{2} & 1a_{1}\\\hline
&  &  & 5\cdot4a_{5} & 4\cdot3a_{4} & 3\cdot2a_{3} & 2\cdot1a_{2}\\\hline
&  &  &  & 5\cdot4\cdot3a_{5} & 4\cdot3\cdot2a_{4} & 3\cdot2\cdot1a_{3}%
\end{array}
\right\vert, \\
D\left(  2,2,1\right)  & = & \operatorname*{dp}\left[
\begin{array}
[c]{l}%
F^{(0)}x^{0}\\\hline
F^{(1)}x^{1}\\
F^{(1)}x^{0}\\\hline
F^{(2)}x^{1}\\
F^{(2)}x^{0}\\\hline
F^{(3)}x^{0}%
\end{array}
\right]  =\dfrac{1}{a_{5}}\left\vert
\begin{array}
[c]{cccccc}%
a_{5} & a_{4} & a_{3} & a_{2} & a_{1} & a_{0}\\\hline
5a_{5} & 4a_{4} & 3a_{3} & 2a_{2} & 1a_{1} & \\
& 5a_{5} & 4a_{4} & 3a_{3} & 2a_{2} & 1a_{1}\\\hline
& 5\cdot4a_{5} & 4\cdot3a_{4} & 3\cdot2a_{3} & 2\cdot1a_{2} & \\
&  & 5\cdot4a_{5} & 4\cdot3a_{4} & 3\cdot2a_{3} & 2\cdot1a_{2}\\\hline
&  &  & 5\cdot4\cdot3a_{5} & 4\cdot3\cdot2a_{4} & 3\cdot2\cdot1a_{3}%
\end{array}
\right\vert,
\\
D\left(  2,1,1,1\right)  & = & \operatorname*{dp}\left[
\begin{array}
[c]{l}%
F^{(0)}x^{0}\\\hline
F^{(1)}x^{1}\\
F^{(1)}x^{0}\\\hline
F^{(2)}x^{0}\\\hline
F^{(3)}x^{0}\\\hline
F^{(4)}x^{0}%
\end{array}
\right]  =\dfrac{1}{a_{5}}\left\vert
\begin{array}
[c]{cccccc}%
a_{5} & a_{4} & a_{3} & a_{2} & a_{1} & a_{0}\\\hline
5a_{5} & 4a_{4} & 3a_{3} & 2a_{2} & 1a_{1} & \\
& 5a_{5} & 4a_{4} & 3a_{3} & 2a_{2} & 1a_{1}\\\hline
&  & 5\cdot4a_{5} & 4\cdot3a_{4} & 3\cdot2a_{3} & 2\cdot1a_{2}\\\hline
&  &  & 5\cdot4\cdot3a_{5} & 4\cdot3\cdot2a_{4} & 3\cdot2\cdot1a_{3}\\\hline
&  &  &  & 5\cdot4\cdot3\cdot2a_{5} & 4\cdot3\cdot2\cdot1a_{4}%
\end{array}
\right\vert, \\
D\left(  1,1,1,1,1\right)  & = & \operatorname*{dp}\left[
\begin{array}
[c]{l}%
F^{(1)}x^{0}\\\hline
F^{(2)}x^{0}\\\hline
F^{(3)}x^{0}\\\hline
F^{(4)}x^{0}\\\hline
F^{\left(  5\right)  }x^{0}%
\end{array}
\right]  =\dfrac{1}{a_{5}}\left\vert
\begin{array}
[c]{cccccc}
{5a_{5}} & {4a_{4}} & {3a_{3}} & {2a_{2}} & {1a_{1} } \\\hline
&  5\cdot4a_{5} & 4\cdot3a_{4} & 3\cdot2a_{3} & 2\cdot1a_{2}\\\hline
&  & 5\cdot4\cdot3a_{5} & 4\cdot3\cdot2a_{4} & 3\cdot2\cdot1a_{3}\\\hline
&  &  & 5\cdot4\cdot3\cdot2a_{5} & 4\cdot3\cdot2\cdot1a_{4}\\\hline
&  &  &  & 5\cdot4\cdot3\cdot2\cdot1a_{5}%
\end{array}
\right\vert.
\end{array}
\]
Note that the last one $D\left(  1,1,1,1,1\right)  =5^{5}4^{4}3^{3}2^{2}%
1^{1}a_{5}^{4}$. Since $a_{5}\neq0$, we see that $D\left(  1,1,1,1,1\right)
\neq0$.
\end{example}

{
To present the main theorem, we recall the following definition for the conjugate of $\boldsymbol{\mu}\in \mathcal{M}(n)$.
}
{
\begin{definition}[Conjugate]
Let $\boldsymbol{\mu}=\left( \mu_{1},\ldots,\mu_{m}\right)\in\mathcal{M}(n) $. Then the \emph{%
conjugate} $\overline{\boldsymbol{\mu}}=\left(\overline{\mu}_1,\ldots,\overline{\mu}_s\right) $ of $\boldsymbol{\mu}$ is defined by%
\begin{align*}
s & =\max_{1\le i\le m}\mu_i=\mu_1, \\
\overline{\mu}_{i} & =\#\left\{ \mu_j:\mu_{j}\geq
i\right\} \ \ \ \text{for }i=1,\ldots,s.
\end{align*}
\end{definition}
}

\begin{theorem}
[Main Result]\label{thm:main_result}Let $F=\sum_{i=0}^{n}a_{i}x^{i}%
\ \ $where$\ \ a_{n}\neq0$. Let $\mathcal{M}(n)=\left\{  \boldsymbol{\mu}%
_{0},\boldsymbol{\mu}_{1},\ldots,\boldsymbol{\mu}_{p}\right\}  $ where the
entries are ordered in the lexicographically
{decreasing order in their conjugates $\overline{\boldsymbol{\mu}_{i}}$'s}.
Then we have the following conditions for the multiplicity vectors.
\[
\operatorname*{mult}(F)=\left\{
\begin{array}
[c]{llll}%
{\boldsymbol{\mu}_{0}}, & \text{if } & {D}\left(
{\overline{\boldsymbol{\mu}_{0}}}\right)  & \neq0;\\
\ \vdots & \vdots & \ {\vdots} & \\
{\boldsymbol{\mu}_{p-1},} & \text{else\ if} & {D}\left(
{\overline{\boldsymbol{\mu}_{p-1}}}\right)  & \neq0;\\
{\boldsymbol{\mu}_{p},} & \text{else\ if} & {D}\left(
{\overline{\boldsymbol{\mu}_{p}}}\right)  & \neq0.
\end{array}
\right.
\]
Equivalently,
\[
\operatorname*{mult}(F)={\boldsymbol{\mu}_{i}}\ \ \ \Longleftrightarrow
\ \ \ D\left(  {\overline{\boldsymbol{\mu}_{0}}}\right)
=\cdots=D\left(  {\overline{\boldsymbol{\mu}_{i-1}}}\right)
=0\wedge D\left(  {\overline{\boldsymbol{\mu}_{i}}}\right)
\neq0.
\]

\end{theorem}

\begin{example}
We have the following condition for each multiplicity vector for degree $5$.
\[
\operatorname*{mult}(F)=\left\{
\begin{array}
[c]{llll}%
\left(  1,1,1,1,1\right),  & \text{if } & D\left(  5\right)  & \neq0;\\
(2,1,1,1), & \text{else\ if} & {D}\left(  4,1\right)  & \neq0;\\
(2,2,1), & \text{else\ if} & {D}\left(  3,2\right)  & \neq0;\\
(3,1,1), & \text{else\ if} & {D}\left(  3,1,1\right)  & \neq0;\\
(3,2), & \text{else\ if} & {D}\left(  2,2,1\right)  & \neq0;\\
\left(  4,1\right),  & \text{else\ if} & {D}\left(  2,1,1,1\right)  & \neq0;\\
\left(  5\right),  & \text{else\ if} & {D}\left(  1,1,1,1,1\right)  & \neq0.
\end{array}
\right.
\]
Equivalently, for instance,%
\[
\operatorname*{mult}(F)=\left(  2,2,1\right)  \ \ \ \Longleftrightarrow
\ \ \ D\left(  5\right)  ={D}\left(  4,1\right)  =0\wedge D\left(  3,2\right)
\neq0.
\]

\end{example}

\begin{remark}Note that ${\overline{\boldsymbol{\mu }_{p}}}=(1,\ldots
,1)$ and
\[
{D}\left(  {\overline{\boldsymbol{\mu }_{p}}}\right)  =\frac
{1}{a_{n}}\left|
\begin{array}
[c]{cccc}%
na_{n} & \cdots & \cdots & 1a_{1}\\
& n\left(  n-1\right)  a_{n} & \cdots & 2\cdot1a_{2}\\
&  & \ddots & \vdots~~~\\
&  &  & n\left(  n-1\right)  \cdots1a_{n}%
\end{array}
\right|  =\prod_{i=1}^{n}i^{i}\cdot a_{n}^{n-1}\neq0.
\]
Hence the last condition is always satisfied and there is no need to check the condition.
\end{remark}

\section{Proof of the Main Theorem}

Here is a high level view of the proof. We start with converting ${D}\left(
\boldsymbol{\mu}\right) $ into the equivalent symmetric polynomials in generic
roots (though displayed as a ratio of two determinants) which is easier to
embed the multiplicity information. Then by making use of the connection
between divided difference with multiple nodes and the derivatives of higher
orders at the nodes, we convert the expression in generic roots to that in
distinct roots with multiplicity information integrated. The theorem will be
proved by eliminating the entries in the determinantal expression obtained
from the second stage which may vanish under the given multiplicity structure.
\label{sec:proof}

\subsection{Multiplicity discriminant in terms of roots}

We first understand what the multiplicity discriminants look like in terms of
roots. .

\begin{notation}
$V(\alpha_{1},\ldots,\alpha_{n}):=\left|
\begin{array}
[c]{ccc}%
\alpha_{1}^{n-1} & \cdots & \alpha_{n}^{n-1}\\
\vdots &  & \vdots\\
\alpha_{1}^{0} & \cdots & \alpha_{n}^{0}%
\end{array}
\right| $.
\end{notation}

\begin{lemma}
[Multiplicity discriminant in generic roots]\label{lem:D_mu_in_roots} Let
$F=a_{n}(x-\alpha_{1})\cdots(x-\alpha_{n})$ and $\boldsymbol{\gamma}%
=(\gamma_{1},\ldots,\gamma_{s})\in\mathcal{M}(n) $. Then
\begin{equation}
D(\boldsymbol{\gamma})=\dfrac{a_{n}^{\gamma_{1}-2}\cdot\left\vert
\begin{array}
[c]{lcl}%
F^{(1)}(\alpha_{1})\alpha_{1}^{\gamma_{1}-1} & \cdots & F^{(1)}(\alpha
_{n})\alpha_{n}^{\gamma_{1}-1}\\
~~~~~~~\vdots &  & ~~~~~~~\vdots\\
F^{(1)}(\alpha_{1})\alpha_{1}^{0} & \cdots & F^{(1)}(\alpha_{n})\alpha_{n}%
^{0}\\\hline
~~~~~~~\vdots &  & ~~~~~~~\vdots\\\hline
F^{(s)}(\alpha_{1})\alpha_{1}^{\gamma_{s}-1} & \cdots & F^{(s)}(\alpha
_{n})\alpha_{n}^{\gamma_{s}-1}\\
~~~~~~~\vdots &  & ~~~~~~~\vdots\\
F^{(s)}(\alpha_{1})\alpha_{1}^{0} & \cdots & F^{(s)}(\alpha_{n})\alpha_{n}^{0}%
\end{array}
\right\vert }{V(\alpha_{1},\ldots,\alpha_{n})}.\label{eqs:discr_in_roots}%
\end{equation}

\end{lemma}

\begin{proof}
\

\begin{enumerate}
\item Since $\gamma_{1}\geq\cdots\geq\gamma_{s}$ and $\gamma_{0}=\gamma_{1}%
-1$, we have
\[
\deg(F^{(0)}x^{n-2})>\cdots>\deg(F^{(0)}x^{\gamma_{1}-1})>\max({\deg(F^{(0)}%
x^{\gamma_{0}-1}),\deg(F^{(1)}x^{\gamma_{1}-1}),\ldots,\deg(F^{(s)}x^{\gamma_{s}-1})}).
\]
Thus
\[
D(\boldsymbol{\gamma})=\frac{1}{a_{n}}\operatorname*{dp}\left[
\begin{array}
[c]{l}%
F^{(0)}x^{\gamma_{1}-2}\\
~~~~\vdots\\
F^{(0)}x^{0}\\\hline
F^{(1)}x^{\gamma_{1}-1}\\
~~~~\vdots\\
F^{(1)}x^{0}\\\hline
~~~~\vdots\\\hline
F^{(s)}x^{\gamma_{s}-1}\\
~~~~\vdots\\
F^{(s)}x^{0}%
\end{array}
\right]  =\frac{1}{a_{n}}\cdot a_{n}^{\gamma_{1}-n}\operatorname*{dp}\left[
\begin{array}
[c]{l}%
F^{(0)}x^{n-2}\\
~~~~\vdots\\
F^{(0)}x^{\gamma_{1}-1}\\
F^{(0)}x^{\gamma_{1}-2}\\
~~~~\vdots\\
F^{(0)}x^{0}\\\hline
F^{(1)}x^{\gamma_{1}-1}\\
~~~~\vdots\\
F^{(1)}x^{0}\\\hline
~~~~\vdots\\\hline
F^{(s)}x^{\gamma_{s}-1}\\
~~~~\vdots\\
F^{(s)}x^{0}%
\end{array}
\right]  =a_{n}^{\gamma_{1}-n-1}\operatorname*{dp}\left[
\begin{array}
[c]{l}%
F^{(0)}x^{n-2}\\
~~~~\vdots\\
F^{(0)}x^{0}\\\hline
F^{(1)}x^{\gamma_{1}-1}\\
~~~~\vdots\\
F^{(1)}x^{0}\\\hline
~~~~\vdots\\\hline
F^{(s)}x^{\gamma_{s}-1}\\
~~~~\vdots\\
F^{(s)}x^{0}%
\end{array}
\right].
\]

\item Now we recall the following result from \cite{2021_Hong_Yang} which is
the key for proving the lemma. Let $G_{1},\ldots,G_{n}\in\mathbb{C}\left[
x\right]  _{2n-2}$ where $\mathbb{C}\left[  x\right]  _{2n-2}$ consists of all
the polynomials in $x$ with degree no greater than $2n-2$. Then%

\begin{equation}
\label{eqs:dp_in_roots}\operatorname*{dp}\left[
\begin{array}
[c]{l}%
F^{(0)}x^{n-2}\\
~~~~\vdots\\
F^{(0)}x^{0}\\
~~~~G_{1}\\
~~~~\vdots\\
~~~~G_{n}%
\end{array}
\right]  =\dfrac{a_{n}^{n-1}\cdot\left\vert
\begin{array}
[c]{ccc}%
G_{1}(\alpha_{1}) & \cdots & G_{1}(\alpha_{n})\\
\vdots &  & \vdots\\
G_{n}(\alpha_{1}) & \cdots & G_{n}\left(  \alpha_{n}\right)
\end{array}
\right\vert }{V(\alpha_{1},\ldots,\alpha_{n})}.%
\end{equation}

\item After specializing $G_{1},\ldots,G_{n}$ in \eqref{eqs:dp_in_roots} with
$F^{(1)}x^{\gamma_{1}-1},\ldots,F^{(1)}x^{0},\ldots,F^{(s)}x^{\gamma_{s}%
-1},\ldots,F^{(s)}x^{0}$, respectively, we have
\[
D(\boldsymbol{\gamma})=a_{n}^{\gamma_{1}-n-1}\cdot\dfrac{a_{n}^{n-1}%
\cdot\left\vert
\begin{array}
[c]{lcl}%
\left(  F^{(1)}x^{\gamma_{1}-1}\right)  (\alpha_{1}) & \cdots & \left(
F^{(1)}x^{\gamma_{1}-1}\right)  (\alpha_{n})\\
~~~~~~~\vdots &  & ~~~~~~~\vdots\\
\left(  F^{(1)}x^{0}\right)  (\alpha_{1}) & \cdots & \left(  F^{(1)}%
x^{0}\right)  (\alpha_{n})\\\hline
~~~~~~~\vdots &  & ~~~~~~~\vdots\\\hline
\left(  F^{(s)}x^{\gamma_{s}-1}\right)  (\alpha_{1}) & \cdots & \left(
F^{(s)}x^{\gamma_{s}-1}\right)  (\alpha_{n})\\
~~~~~~~\vdots &  & ~~~~~~~\vdots\\
\left(  F^{(s)}x^{0}\right)  (\alpha_{1}) & \cdots & \left(  F^{(s)}%
x^{0}\right)  (\alpha_{n})
\end{array}
\right\vert }{V(\alpha_{1},\ldots,\alpha_{n})}%
\]
which can be easily simplified into \eqref{eqs:discr_in_roots}.
\end{enumerate}
\end{proof}

\begin{remark}
It is very important to note that the right hand side is a polynomial function
in $\alpha_{1},\ldots,\alpha_{n}$, even though written as a rational function,
since the numerator is exactly divisible by the denominator. Hence the above
definition should be read as follows:

\begin{enumerate}
\item Treating $\alpha_{1},\ldots,\alpha_{n}$ as distinct indeterminates,
carry out the exact division obtaining a polynomial.

\item Treating $\alpha_{1},\ldots,\alpha_{n}$ as numbers, evaluate the
resulting polynomial.
\end{enumerate}
\end{remark}

\begin{lemma}
[Multiplicity discriminant in multiple roots]\label{lem:D_mu_in_mroots} Let
$F$ be of degree $n$ with $m$ distinct roots $r_{1},\ldots,r_{m}$, of
multiplicities $\mu_{1},\ldots,\mu_{m}$, that is $\mu_{1}+\cdots+\mu_{m}=n$.
Let $\boldsymbol{\gamma}=\left(  \gamma_{1},\ldots,\gamma_{s}\right)
\in\Gamma(n)$. Then we have
\begin{equation}
D(\boldsymbol{\gamma})=\frac{c\cdot\left\vert \setlength{\arraycolsep}{2.0pt}%
\begin{array}
[c]{lcl|c|lcl}%
({F^{(1)}x^{\gamma_{1}-1}})^{(0)}(r_{1}) & \cdots & ({F^{(1)}x^{\gamma_{1}-1}%
})^{(\mu_{1}-1)}(r_{1}) & \cdots\cdots & ({F^{(1)}x^{\gamma_{1}-1}}%
)^{(0)}(r_{m}) & \cdots & ({F^{(1)}x^{\gamma_{1}-1}})^{(\mu_{m}-1)}(r_{m})\\
~~~~~~~~\vdots &  & ~~~~~~~~\vdots &  & ~~~~~~~~\vdots &  & ~~~~~~~~\vdots\\
({F^{(1)}x^{0}})^{(0)}(r_{1}) & \cdots & ({F^{(1)}x^{0}})^{(\mu_{1}-1)}%
(r_{1}) & \cdots\cdots & ({F^{(1)}x^{0}})^{(0)}(r_{m}) & \cdots &
({F^{(1)}x^{0}})^{(\mu_{m}-1)}(r_{m})\\\hline
~~~~~~~~\vdots &  & ~~~~~~~~\vdots &  & ~~~~~~~~\vdots &  & ~~~~~~~~\vdots
\\\hline
({F^{(s)}}x^{\gamma_{s}-1})^{(0)}(r_{1}) & \cdots & ({F}^{(s)}x^{\gamma_{s}%
-1})^{(\mu_{1}-1)}(r_{1}) & \cdots\cdots & ({F^{(s)}}x^{\gamma_{s}-1}%
)^{(0)}(r_{m}) & \cdots & ({F}^{(s)}x^{\gamma_{s}-1})^{(\mu_{m}-1)}(r_{m})\\
~~~~~~~~\vdots &  & ~~~~~~~~\vdots &  & ~~~~~~~~\vdots &  & ~~~~~~~~\vdots\\
({F^{(s)}}x^{0})^{(0)}(r_{1}) & \cdots & ({F^{(s)}}x^{0})^{(\mu_{1}-1)}%
(r_{1}) & \cdots\cdots & ({F^{(s)}}x^{0})^{(0)}(r_{m}) & \cdots & ({F^{(s)}%
}x^{0})^{(\mu_{m}-1)}(r_{m})
\end{array}
\right\vert }{\prod\limits_{1\leq i<j\leq m}(r_{i}-r_{j})^{\mu_{i}\mu_{j}}%
}\label{eq:div_diff_det}%
\end{equation}
where $c=\pm1\Big/\left(  \prod_{i=1}^{m}\prod_{j=0}^{\mu_{i}-1}j!\right)
\cdot a_{n}^{\gamma_{1}-2}$.
\end{lemma}

\begin{proof}
\

\begin{enumerate}
\item Let $F=a_{n}(x-\alpha_{1})\cdots(x-\alpha_{n})$. When $\alpha_{1}%
,\ldots,\alpha_{n}$ are treated as numbers, without loss of generality, we may
assume that $\alpha_{1},\ldots,\alpha_{n}$ are grouped into $m$ sets as
follows:
\[%
\begin{array}
[c]{llccccc}%
\mathcal{S}_{1}:= & \{\alpha_{1} & \cdots & \cdots & \cdots & \cdots &
\alpha_{\mu_{1}}\},\\
\mathcal{S}_{2}:= & \{\alpha_{\mu_{1}+1} & \cdots & \cdots & \cdots &
\alpha_{\mu_{1}+\mu_{2}}\}, & \\
\qquad\vdots &  &  &  &  &  & \\
\mathcal{S}_{m}:= & \{\alpha_{\mu_{1}+\cdots+\mu_{m-1}+1} & \cdots &
\alpha_{\mu_{1}+\cdots+\mu_{m-1}+\mu_{m}}\}. &  &  &
\end{array}
\]
where elements in $\mathcal{S}_{i}$ are all equal to $r_{i}$.

\item Recall that
\[
{D}(\boldsymbol{\gamma})=a_{n}^{\gamma_{1}-2}\cdot{\left\vert
\begin{array}
[c]{lcl}%
({F}^{(1)}x^{\gamma_{1}-1})(\alpha_{1}) & \cdots & ({F}^{(1)}x^{\gamma_{1}%
-1})(\alpha_{n})\\
~~~~~~\vdots &  & ~~~~~~\vdots\\
({F}^{(1)}x^{0})(\alpha_{1}) & \cdots & ({F}^{(1)}x^{0})(\alpha_{n})\\\hline
~~~~~~\vdots &  & ~~~~~~\vdots\\
~~~~~~\vdots &  & ~~~~~~\vdots\\\hline
({F}^{(s)}x^{\gamma_{s}-1})(\alpha_{1}) & \cdots & ({F}^{(s)}x^{\gamma_{s}%
-1})(\alpha_{n})\\
~~~~~~\vdots &  & ~~~~~~\vdots\\
({F}^{(s)}x^{0})(\alpha_{1}) & \cdots & ({F}^{(s)}x^{0})(\alpha_{n})
\end{array}
\right\vert \big/V(\alpha_{1},\ldots,\alpha_{n}).}%
\]
Next we will treat $\alpha_{1},\ldots,\alpha_{n}$ as indeterminates and carry
out the exact division so that difference between the collapsed $\alpha_{i}$'s
do not appear in the denominator.

\item For the sake of simplicity, we use the follow shorthand notion:
\[
\boldsymbol{F}:=\left[  F^{(1)}x^{\gamma_{1}-1},\ldots,F^{(1)}x^{0}%
,\ldots,F^{(s)}x^{\gamma_{s}-1},\ldots,F^{(s)}x^{0}\right]  ^{T}.%
\]

\item Let $P[x_{1},\ldots,x_{i}]$ denote the $(i-1)$th divided difference of
$P\in\mathbb{C}[x]$ at $x_{1},\ldots,x_{i}$ {defined recursively as follows:
\[
P[x_{1},\ldots,x_{i}]=\left\{
\begin{array}{ll}
P(x_1),&\text{if}\ i=1;\\[3pt]
\dfrac{P[x_1,\ldots,x_{i-2},x_i]-P[x_1,\ldots,x_{i-2},x_{i-1}]}{x_i-x_{i-1}},&\text{if}\ i>1.
\end{array}
\right.
\]}
Let
\begin{align*}
\boldsymbol{F}[\alpha_{1},\ldots,\alpha_{i}]:= &  \left[  (F^{(1)}%
x^{\gamma_{1}-1})[\alpha_{1},\ldots,\alpha_{i}],\ldots,(F^{(1)}x^{0}%
)[\alpha_{1},\ldots,\alpha_{i}],\right. \\
&  \left.  \quad\ldots\ldots,(F^{(s)}x^{\gamma_{s}-1})[\alpha_{1}%
,\ldots,\alpha_{i}],\ldots,(F^{(s)}x^{0})[\alpha_{1},\ldots,\alpha
_{i}]\right]  ^{T}.%
\end{align*}

\item It follows that
\begin{align*}
D(\boldsymbol{\gamma}) &  =a_{n}^{\gamma_{1}-2}\cdot\frac{\left\vert
\begin{array}
[c]{ccc}%
\boldsymbol{F}(\alpha_{1}) & \cdots & \boldsymbol{F}(\alpha_{n})
\end{array}
\right\vert }{V(\alpha_{1},\ldots,\alpha_{n})}\\[5pt]
&  =a_{n}^{\gamma_{1}-2}\cdot\frac{\left\vert
\begin{array}
[c]{cccccc}%
\boldsymbol{F}[\alpha_{1}] & \cdots & \boldsymbol{F}[\alpha_{\mu_{1}}] &
\boldsymbol{F}[\alpha_{\mu_{1}+1}] & \cdots & \boldsymbol{F}[\alpha_{n}]
\end{array}
\right\vert }{\prod\limits_{\substack{\alpha_{i},\alpha_{j}\in\mathcal{S}_{1}
\\j-i>0}}(\alpha_{i}-\alpha_{j})\prod\limits_{\substack{\alpha_{i},\alpha
_{j}\notin\mathcal{S}_{1} \\j-i>0}}(\alpha_{i}-\alpha_{j})\prod
\limits_{\substack{\alpha_{i}\in\mathcal{S}_{1} \\\alpha_{j}\notin%
\mathcal{S}_{1}}}(\alpha_{i}-\alpha_{j})}\\[8pt]
&  =\pm a_{n}^{\gamma_{1}-2}\cdot\frac{\left\vert
\begin{array}
[c]{ccccccc}%
\boldsymbol{F}[\alpha_{1}] & \boldsymbol{F}[\alpha_{1},\alpha_{2}] & \cdots &
\boldsymbol{F}[\alpha_{\mu_{1}-1},\alpha_{\mu_{1}}] & \boldsymbol{F}%
[\alpha_{\mu_{1}+1}] & \cdots & \boldsymbol{F}[\alpha_{n}]
\end{array}
\right\vert }{\prod\limits_{\substack{\alpha_{i},\alpha_{j}\in\mathcal{S}_{1}
\\j-i>1}}(\alpha_{i}-\alpha_{j})\prod\limits_{\substack{\alpha_{i},\alpha
_{j}\notin\mathcal{S}_{1} \\j-i>0}}(\alpha_{i}-\alpha_{j})\prod
\limits_{\substack{\alpha_{i}\in\mathcal{S}_{1} \\\alpha_{j}\notin%
\mathcal{S}_{1}}}(\alpha_{i}-\alpha_{j})}\\[10pt]
&  =\pm a_{n}^{\gamma_{1}-2}\cdot\frac{\left\vert
\begin{array}
[c]{cccccccc}%
\boldsymbol{F}[\alpha_{1}] & \boldsymbol{F}[\alpha_{1},\alpha_{2}] &
\boldsymbol{F}[\alpha_{1},\alpha_{2},\alpha_{3}] & \cdots & \boldsymbol{F}%
[\alpha_{\mu_{1}-2},\alpha_{\mu_{1}-1},\alpha_{\mu_{1}}] & \boldsymbol{F}%
[\alpha_{\mu_{1}+1}] & \cdots & \boldsymbol{F}[\alpha_{n}]
\end{array}
\right\vert }{\prod\limits_{\substack{\alpha_{i},\alpha_{j}\in\mathcal{S}_{1}
\\j-i>2}}(\alpha_{i}-\alpha_{j})\prod\limits_{\substack{\alpha_{i},\alpha
_{j}\notin\mathcal{S}_{1} \\j-i>0}}(\alpha_{i}-\alpha_{j})\prod
\limits_{\substack{\alpha_{i}\in\mathcal{S}_{1} \\\alpha_{j}\notin%
\mathcal{S}_{1}}}(\alpha_{i}-\alpha_{j})}\\
&  ~\vdots\\
&  ~\\
&  =\pm a_{n}^{\gamma_{1}-2}\cdot\frac{\left\vert
\begin{array}
[c]{ccccccc}%
\boldsymbol{F}[\alpha_{1}] & \boldsymbol{F}[\alpha_{1},\alpha_{2}] & \cdots &
\boldsymbol{F}[\alpha_{1},\ldots,\alpha_{\mu_{1}}] & \boldsymbol{F}%
(\alpha_{\mu_{1}+1}) & \cdots & \boldsymbol{F}(\alpha_{n})
\end{array}
\right\vert }{\prod\limits_{\substack{\alpha_{i},\alpha_{j}\notin%
\mathcal{S}_{1} \\j-i>0}}(\alpha_{i}-\alpha_{j})\prod\limits_{\substack{\alpha
_{i}\in\mathcal{S}_{1} \\\alpha_{j}\notin\mathcal{S}_{1}}}(\alpha_{i}%
-\alpha_{j})}.%
\end{align*}

\item Repeating the procedure for $\alpha_{j}$'s in each $\mathcal{S}_{i}$ for
$i=2,\ldots,m$ successively, we get
\[
D(\boldsymbol{\gamma})=\pm a_{n}^{\gamma_{1}-2}\cdot\frac{\left\vert
\begin{array}
[c]{ccc|cc|ccc}%
\boldsymbol{F}[\alpha_{1}] & \cdots & \boldsymbol{F}[\alpha_{1},\ldots
,\alpha_{\mu_{1}}] & \cdots & \cdots & \boldsymbol{F}[\alpha_{\mu_{1}%
+\cdots+\mu_{m-1}+1}] & \cdots & \boldsymbol{F}[\alpha_{\mu_{1}+\cdots
+\mu_{m-1}+1},\ldots,\alpha_{n}]
\end{array}
\right\vert }{\prod\limits_{1\leq i<j\leq m}\prod\limits_{\substack{\alpha
_{p}\in\mathcal{S}_{i} \\\alpha_{q}\in\mathcal{S}_{j}}}(\alpha_{p}-\alpha
_{q})}.%
\]

\item Now we substitute $\alpha_{1}=\cdots=\alpha_{\mu_{1}}=r_{1}$, \ldots,
$\alpha_{\mu_{1}+\cdots+\mu_{m-1}+1}=\cdots=\alpha_{n}=r_{m}$ into
$D(\boldsymbol{\gamma})$ and obtain
\begin{equation}
\label{eq:D_in_ddiv}D(\boldsymbol{\gamma})= \pm a_{n}^{\gamma_{1}-2}\cdot
\frac{\left\vert
\begin{array}
[c]{ccc|cc|ccc}%
\boldsymbol{F}[r_{1}] & \cdots & \boldsymbol{F}[r_{1},\ldots,r_{1}] & \cdots &
\cdots & \boldsymbol{F}[r_{m}] & \cdots & \boldsymbol{F}[r_{m},\ldots,r_{m}]
\end{array}
\right\vert }{\prod\limits_{1\le i<j\le m}(r_{i}-r_{j})^{\mu_{i}\mu_{j}}}.%
\end{equation}

\item {By \cite[Equation (2.1.5.7a)]{2002_Stoer_Bulirsch},} for any given polynomial $P\in\mathbb{C}[x]$,
\[
P[\underbrace{r_{i},\ldots,r_{i}}_{k\ r_{i}\text{'s}}]=\frac{P^{(k-1)}(r_{i}%
)}{(k-1)!}.%
\]
Hence
\begin{equation}
\label{eq:ddiv}\boldsymbol{F}[\underbrace{r_{i},\ldots,r_{i}}_{k\ r_{i}%
\text{'s}}]= \left[  \frac{({F^{(1)}x^{\gamma_{1}-1}})^{(k-1)}(r_{i})}%
{(k-1)!},\ldots,\frac{({F^{(1)}x^{0}})^{(k-1)}(r_{i})}{(k-1)!},\ldots
,\frac{({F}^{(s)}x^{\gamma_{s}-1})^{(k-1)}(r_{i})}{(k-1)!},\ldots,\frac
{({F}^{(s)}x^{0})^{(k-1)}(r_{i})}{(k-1)!}\right]  ^{T}.%
\end{equation}

\item Substituting \eqref{eq:ddiv} into \eqref{eq:D_in_ddiv} , we have
\begin{align*}
D(\boldsymbol{\gamma}) &  =\pm a_{n}^{\gamma_{1}-2}\cdot\frac{\left|
\setlength{\arraycolsep}{2.5pt}%
\begin{array}
[c]{lcl|c|lcl}%
\frac{({F^{(1)}x^{\gamma_{1}-1}})^{(0)}(r_{1})}{0!} & \cdots & \frac
{({F^{(1)}x^{\gamma_{1}-1}})^{(\mu_{1}-1)}(r_{1})}{(\mu_{1}-1)!} &
\cdots\cdots & \frac{({F^{(1)}x^{\gamma_{1}-1}})^{(0)}(r_{m})}{0!} & \cdots &
\frac{({F^{(1)}x^{\gamma_{1}-1}})^{(\mu_{m}-1)}(r_{m})}{(\mu_{m}-1)!}\\
~~~~~~~~\vdots &  & ~~~~~~~~\vdots &  & ~~~~~~~~\vdots &  & ~~~~~~~~\vdots\\
\frac{({F^{(1)}x^{0}})^{(0)}(r_{1})}{0!} & \cdots & \frac{({F^{(1)}x^{0}%
})^{(\mu_{1}-1)}(r_{1})}{(\mu_{1}-1)!} & \cdots\cdots & \frac{({F^{(1)}x^{0}%
})^{(0)}(r_{m})}{0!} & \cdots & \frac{({F^{(1)}x^{0}})^{(\mu_{m}-1)}(r_{m}%
)}{(\mu_{m}-1)!}\\\hline
~~~~~~~~\vdots &  & ~~~~~~~~\vdots &  & ~~~~~~~~\vdots &  & ~~~~~~~~\vdots
\\\hline
\frac{({F}^{(s)}x^{\gamma_{s}-1})^{(0)}(r_{1})}{0!} & \cdots & \frac
{({F}^{(s)}x^{\gamma_{s}-1})^{(\mu_{1}-1)}(r_{1})}{(\mu_{1}-1)!} &
\cdots\cdots & \frac{({F}^{(s)}x^{\gamma_{s}-1})^{(0)}(r_{m})}{0!} & \cdots &
\frac{({F^{(s)}}x^{\gamma_{s}-1})^{(\mu_{m}-1)}(r_{m})}{(\mu_{m}-1)!}\\
~~~~~~~~\vdots &  & ~~~~~~~~\vdots &  & ~~~~~~~~\vdots &  & ~~~~~~~~\vdots\\
\frac{({F}^{(s)}x^{0})^{(0)}(r_{1})}{0!} & \cdots & \frac{({F^{(s)}}%
x^{0})^{(\mu_{1}-1)}(r_{1})}{(\mu_{1}-1)!} & \cdots\cdots & \frac{({F^{(s)}%
}x^{0})^{(0)}(r_{m})}{0!} & \cdots & \frac{({F^{(s)}}x^{0})^{(\mu_{m}%
-1)}(r_{m})}{(\mu_{m}-1)!}%
\end{array}
\right|  }{\prod\limits_{1\leq i<j\leq m}(r_{i}-r_{j})^{\mu_{i}\mu_{j}}}\\
&  =\frac{c\cdot\left| \setlength{\arraycolsep}{2.0pt}%
\begin{array}
[c]{lcl|c|lcl}%
({F^{(1)}x^{\gamma_{1}-1}})^{(0)}(r_{1}) & \cdots & ({F^{(1)}x^{\gamma_{1}-1}%
})^{(\mu_{1}-1)}(r_{1}) & \cdots\cdots & ({F^{(1)}x^{\gamma_{1}-1}}%
)^{(0)}(r_{m}) & \cdots & ({F^{(1)}x^{\gamma_{1}-1}})^{(\mu_{m}-1)}(r_{m})\\
~~~~~~~~\vdots &  & ~~~~~~~~\vdots &  & ~~~~~~~~\vdots &  & ~~~~~~~~\vdots\\
({F^{(1)}x^{0}})^{(0)}(r_{1}) & \cdots & ({F^{(1)}x^{0}})^{(\mu_{1}-1)}%
(r_{1}) & \cdots\cdots & ({F^{(1)}x^{0}})^{(0)}(r_{m}) & \cdots &
({F^{(1)}x^{0}})^{(\mu_{m}-1)}(r_{m})\\\hline
~~~~~~~~\vdots &  & ~~~~~~~~\vdots &  & ~~~~~~~~\vdots &  & ~~~~~~~~\vdots
\\\hline
({F^{(s)}}x^{\gamma_{s}-1})^{(0)}(r_{1}) & \cdots & ({F}^{(s)}x^{\gamma_{s}%
-1})^{(\mu_{1}-1)}(r_{1}) & \cdots\cdots & ({F^{(s)}}x^{\gamma_{s}-1}%
)^{(0)}(r_{m}) & \cdots & ({F}^{(s)}x^{\gamma_{s}-1})^{(\mu_{m}-1)}(r_{m})\\
~~~~~~~~\vdots &  & ~~~~~~~~\vdots &  & ~~~~~~~~\vdots &  & ~~~~~~~~\vdots\\
({F^{(s)}}x^{0})^{(0)}(r_{1}) & \cdots & ({F^{(s)}}x^{0})^{(\mu_{1}-1)}%
(r_{1}) & \cdots\cdots & ({F^{(s)}}x^{0})^{(0)}(r_{m}) & \cdots & ({F^{(s)}%
}x^{0})^{(\mu_{m}-1)}(r_{m})
\end{array}
\right|  }{\prod\limits_{1\leq i<j\leq m}(r_{i}-r_{j})^{\mu_{i}\mu_{j}}}%
\end{align*}
where $c=\pm1\Big/\left(  \prod_{i=1}^{m}\prod_{j=0}^{\mu_{i}-1}j!\right)
\cdot a_{n}^{\gamma_{1}-2}$.
\end{enumerate}
\end{proof}

\subsection{Connection between the multiplicity discriminants and multiplicity
vectors}

By decompiling Theorem \ref{thm:main_result}, we identify the two essential
ingredients therein, which are re-stated as Lemmas \ref{lem:D_mu_not_zero} and
\ref{lem:D_lambda_all_zeros} below. From now on, we will use $\overline
{\boldsymbol{\gamma}}$ to denote the conjugate of $\boldsymbol{\gamma}%
\in\mathcal{M}(n)$. {To prove the lemmas, we recall the following well known fact \cite{2011_Bona} which depicts the connection between $\gamma$ and its conjugate.}

\begin{lemma}
{Let $\gamma\in\mathcal{M}(n)$. Then $\overline{\overline{\gamma}}=\gamma$. Moreover, if $\gamma$ and $\mu$ are conjugates to each other, then $\#\{\mu_j:\mu_j\ge i\}=\gamma_i$ and $\#\{\lambda_j:\lambda_j\ge i\}=\mu_i$}.
\end{lemma}

\begin{lemma}
\label{lem:D_mu_not_zero} Let $\operatorname*{mult}(F)=\boldsymbol{\mu}$. Then
$D(\bar{\boldsymbol{\mu}})\neq0$.
\end{lemma}

\begin{proof}
In order to convey the main underlying ideas effectively, we will show the
proof for a particular case first. After that, we will generalize the ideas to
arbitrary cases.

\bigskip

\noindent\textbf{Particular case:} Consider the case $n=5$ and
$\operatorname*{mult}(F)=\boldsymbol{\mu}=\left(  3,2\right) $.

\begin{enumerate}
\item Assume that $r_{1}$ and $r_{2}$ are the two distinct roots with
multiplicities $3$ and $2,$ respectively. In other words, $F=a_{5}%
(x-r_{1})^{3}(x-r_{2})^{2}$.

\item Let $\boldsymbol{\gamma}=\bar{\boldsymbol{\mu}}$. Then
\[
\gamma_{1}=\#\{\mu_{j}:\,\mu_{j}\ge1\}=2,\quad\gamma_{2}=\#\{\mu_{j}:\,\mu
_{j}\ge2\}=2,\quad\gamma_{3}=\#\{\mu_{j}:\,\mu_{j}\ge3\}=1.
\]
Thus $\boldsymbol{\gamma}=(2,2,1)$.

\item By Lemma \ref{lem:D_mu_in_mroots},
\begin{align*}
D(\boldsymbol{\gamma})  & =\frac{c\cdot\left|
\setlength{\arraycolsep}{8.0pt}
\begin{array}
[c]{lcl|cc}%
({F^{(1)}x^{1}})^{(0)}(r_{1}) & ({F^{(1)}x^{1}})^{(1)}(r_{1}) & ({F^{(1)}%
x^{1}})^{(2)}(r_{1}) & ({F^{(1)}x^{1}})^{(0)}(r_{2}) & ({F^{(1)}x^{1}}%
)^{(1)}(r_{2})\\
({F^{(1)}x^{0}})^{(0)}(r_{1}) & ({F^{(1)}x^{0}})^{(1)}(r_{1}) & ({F^{(1)}%
x^{0}})^{(2)}(r_{1}) & ({F^{(1)}x^{0}})^{(0)}(r_{2}) & ({F^{(1)}x^{0}}%
)^{(1)}(r_{2})\\\hline
({F^{(2)}x^{1}})^{(0)}(r_{1}) & ({F^{(2)}x^{1}})^{(1)}(r_{1}) & ({F^{(2)}%
x^{1}})^{(2)}(r_{1}) & ({F^{(2)}x^{1}})^{(0)}(r_{2}) & ({F^{(2)}x^{1}}%
)^{(1)}(r_{2})\\
({F^{(2)}x^{0}})^{(0)}(r_{1}) & ({F^{(2)}x^{0}})^{(1)}(r_{1}) & ({F^{(2)}%
x^{0}})^{(2)}(r_{1}) & ({F^{(2)}x^{0}})^{(0)}(r_{2}) & ({F^{(2)}x^{0}}%
)^{(1)}(r_{2})\\\hline
({F^{(3)}x^{0}})^{(0)}(r_{1}) & ({F^{(3)}x^{0}})^{(1)}(r_{1}) & ({F^{(3)}%
x^{0}})^{(2)}(r_{1}) & ({F^{(3)}x^{0}})^{(0)}(r_{2}) & ({F^{(3)}x^{0}}%
)^{(1)}(r_{2})
\end{array}
\right|  }{(r_{1}-r_{2})^{6}}%
\end{align*}
where
\[
c=\pm1\big/\left[ (0!\cdot1!\cdot2!)\cdot(0!\cdot1!)\right] \cdot a_{5}%
^{0}=\pm1/2.
\]

\item Since
\begin{equation}
\label{eqs:derivatives_of_ri}F^{(i)}(r_{1})\left\{
\begin{array}
[c]{ll}%
=0, & \text{for}\ i=0,1,2;\\
\ne0, & \text{for}\ i=3;
\end{array}
\right. \qquad F^{(i)}(r_{2})\left\{
\begin{array}
[c]{ll}%
=0, & \text{for}\ i=0,1;\\
\ne0, & \text{for}\ i=2,
\end{array}
\right.
\end{equation}
{by the Leibniz's rule for derivatives}, we immediately know {\small
\[
\setlength{\arraycolsep}{1.0pt}
\begin{array}
[c]{l}%
({F^{(1)}x^{1}})^{(0)}(r_{1})=0,\ ({F^{(1)}x^{1}})^{(1)}(r_{1}%
)=0,\ ({F^{(1)}x^{1}})^{(2)}(r_{1})=F^{(3)}(r_{1})r_{1}^{1},\ ({F^{(1)}%
x^{1}})^{(0)}(r_{2})=0,\ ({F^{(1)}x^{1}})^{(1)}(r_{2})=F^{(2)}(r_{2}%
)r_{2}^{1},\\[2pt]%
({F^{(1)}x^{0}})^{(0)}(r_{1})=0,\ ({F^{(1)}x^{0}})^{(1)}(r_{1}%
)=0,\ ({F^{(1)}x^{0}})^{(2)}(r_{1})=F^{(3)}(r_{1})r_{1}^{0},\ ({F^{(1)}%
x^{0}})^{(0)}(r_{2})=0,\ ({F^{(1)}x^{0}})^{(1)}(r_{2})=F^{(2)}(r_{2}%
)r_{2}^{0},\\[2pt]%
({F^{(2)}x^{1}})^{(0)}(r_{1})=0,\ ({F^{(2)}x^{1}})^{(1)}(r_{1}%
)=F^{(3)}(r_{1})r_{1}^{1}, \hspace{9.5em} ({F^{(2)}x^{1}})^{(0)}(r_{2}%
)={F^{(2)}(r_{2})r_{2}^{1}},\\[2pt]%
({F^{(2)}x^{0}})^{(0)}(r_{1})=0,\ ({F^{(2)}x^{0}})^{(1)}(r_{1})=F^{(3)}%
(r_{1})r_{1}^{0}, \hspace{9.5em} ({F^{(2)}x^{0}})^{(0)}(r_{2})={F^{(2)}%
(r_{2})r_{2}^{0}},\\[2pt]%
({F^{(3)}x^{0}})^{(0)}(r_{1})=F^{(3)}(r_{1})r_{1}^{0}.%
\end{array}
\]
}

\item Therefore,
\[
D(\boldsymbol{\gamma}) =\frac{c\cdot\left|
\begin{array}
[c]{ccc|cc}%
0 & 0 & F^{(3)}(r_{1})r_{1}^{1} & 0 & F^{(2)}(r_{2})r_{2}^{1}\\
0 & 0 & F^{(3)}(r_{1})r_{1}^{0} & 0 & F^{(2)}(r_{2})r_{2}^{0}\\\hline
0 & F^{(3)}(r_{1})r_{1}^{1} & \cdot & F^{(2)}(r_{2})r_{2}^{1} & \cdot\\
0 & F^{(3)}(r_{1})r_{1}^{0} & \cdot & F^{(2)}(r_{2})r_{2}^{0} & \cdot\\\hline
F^{(3)}(r_{1})r_{1}^{0} & \cdot & \cdot & \cdot & \cdot
\end{array}
\right|  }{(r_{1}-r_{2})^{6}}.%
\]

\item By rearranging the columns of the determinant in the numerator, we have
\begin{align*}
D(\boldsymbol{\gamma})  & =\pm\frac{c\cdot\left|
\begin{array}
[c]{c|cc|cc}
&  &  & F^{(3)}(r_{1})r_{1}^{1} & F^{(2)}(r_{2})r_{2}^{1}\\
&  &  & F^{(3)}(r_{1})r_{1}^{0} & F^{(2)}(r_{2})r_{2}^{0}\\\hline
& F^{(3)}(r_{1})r_{1}^{1} & F^{(2)}(r_{2})r_{2}^{1} & \cdot & \cdot\\
& F^{(3)}(r_{1})r_{1}^{0} & F^{(2)}(r_{2})r_{2}^{0} & \cdot & \cdot\\\hline
F^{(3)}(r_{1})r_{1}^{0} & \cdot & \cdot & \cdot & \cdot
\end{array}
\right|  }{(r_{1}-r_{2})^{6}}\\
& =\pm\frac{c\cdot\left|
\begin{array}
[c]{ccc}
&  & M_{1}\\
& M_{2} & \cdot\\
M_{3} & \cdot & \cdot
\end{array}
\right|  }{(r_{1}-r_{2})^{6}}%
\end{align*}

where%
\[
M_{1}= \left[
\begin{array}
[c]{cc}%
F^{(3)}(r_{1})r_{1}^{1} & F^{(2)}(r_{2})r_{2}^{1}\\
F^{(3)}(r_{1})r_{1}^{0} & F^{(2)}(r_{2})r_{2}^{0}%
\end{array}
\right],  \qquad M_{2}=\left[
\begin{array}
[c]{cc}%
F^{(3)}(r_{1})r_{1}^{1} & F^{(2)}(r_{2})r_{2}^{1}\\
F^{(3)}(r_{1})r_{1}^{0} & F^{(2)}(r_{2})r_{2}^{0}%
\end{array}
\right],  \qquad M_{3}=\left[ F^{(3)}(r_{1})r_{1}^{0}\right].
\]

\item Obviously,
\[
D(\boldsymbol{\gamma})=\pm\frac{c\cdot\left| M_{1}\right| \cdot\left|
M_{2}\right| \cdot\left| M_{3}\right| }{(r_{1}-r_{2})^{6}}.%
\]

We only need to show that $M_{i}\ne0$ for $i=1,2,3$. The claim follows from
the following observations:
\begin{align*}
\left| M_{1}\right|  & =F^{(3)}(r_{1})F^{(2)}(r_{2})V(r_{1},r_{2})\ne0,\\
\left| M_{2}\right|  & =F^{(3)}(r_{1})F^{(2)}(r_{2})V(r_{1},r_{2})\ne0,\\
\left| M_{3}\right|  & =F^{(3)}(r_{1})V(r_{1})\ne0.
\end{align*}
The proof is completed.
\end{enumerate}

\medskip\noindent\textbf{Arbitrary case.} Now we generalize the above ideas to
arbitrary cases.

\begin{enumerate}
\item Let $\boldsymbol{\mu}=(\mu_{1},\ldots,\mu_{m})$. Assume that
$r_{1},\ldots,r_{m}$ are the $m$ distinct roots with multiplicities $\mu
_{1},\ldots,\mu_{m} $ respectively. In other words, $F=a_{n}(x-r_{1})^{\mu
_{1}}\cdots(x-r_{m})^{\mu_{m}}$.

\item Let $\boldsymbol{\gamma}=\bar{\boldsymbol{\mu}}=(\gamma_{1}%
,\ldots,\gamma_{s})$, i.e., $\gamma_{i}=\#\{\mu_{j}:\,\mu_{j}\ge i\}$. Note
that $\gamma_{1}=m$ and $s=\mu_{1}$ since $\mu_{1}\ge\cdots\ge\mu_{m}\ge1$.

\item Recall that
\begin{equation}
\label{eq:rep_D_mu}D(\boldsymbol{\gamma})=\frac{c\cdot\left|
\setlength{\arraycolsep}{1.0pt}%
\begin{array}
[c]{lcl|c|lcl}%
({F^{(1)}x^{\gamma_{1}-1}})^{(0)}(r_{1}) & \cdots & ({F^{(1)}x^{\gamma_{1}-1}%
})^{(\mu_{1}-1)}(r_{1}) & \cdots\cdots & ({F^{(1)}x^{\gamma_{1}-1}}%
)^{(0)}(r_{m}) & \cdots & ({F^{(1)}x^{\gamma_{1}-1}})^{(\mu_{m}-1)}(r_{m})\\
~~~~~~~~\vdots &  & ~~~~~~~~\vdots &  & ~~~~~~~~\vdots &  & ~~~~~~~~\vdots\\
({F^{(1)}x^{0}})^{(0)}(r_{1}) & \cdots & ({F^{(1)}x^{0}})^{(\mu_{1}-1)}%
(r_{1}) & \cdots\cdots & ({F^{(1)}x^{0}})^{(0)}(r_{m}) & \cdots &
({F^{(1)}x^{0}})^{(\mu_{m}-1)}(r_{m})\\\hline
~~~~~~~~\vdots &  & ~~~~~~~~\vdots &  & ~~~~~~~~\vdots &  & ~~~~~~~~\vdots
\\\hline
({F^{(\mu_{1})}}x^{\gamma_{\mu_{1}}-1})^{(0)}(r_{1}) & \cdots & ({F}^{(\mu
_{1})}x^{\gamma_{\mu_{1}}-1})^{(\mu_{1}-1)}(r_{1}) & \cdots\cdots &
({F^{(\mu_{1})}}x^{\gamma_{\mu_{1}}-1})^{(0)}(r_{m}) & \cdots & ({F}^{(\mu
_{1})}x^{\gamma_{\mu_{1}}-1})^{(\mu_{m}-1)}(r_{m})\\
~~~~~~~~\vdots &  & ~~~~~~~~\vdots &  & ~~~~~~~~\vdots &  & ~~~~~~~~\vdots\\
({F^{(\mu_{1})}}x^{0})^{(0)}(r_{1}) & \cdots & ({F^{(\mu_{1})}}x^{0}%
)^{(\mu_{1}-1)}(r_{1}) & \cdots\cdots & ({F^{(\mu_{1})}}x^{0})^{(0)}(r_{m}) &
\cdots & ({F^{(\mu_{1})}}x^{0})^{(\mu_{m}-1)}(r_{m})
\end{array}
\right|  }{\prod\limits_{1\leq i<j\leq m}(r_{i}-r_{j})^{\mu_{i}\mu_{j}}}%
\end{equation}
where $c=\pm1\Big/\left( \prod_{i=1}^{m}\prod_{j=0}^{\mu_{i}-1}j!\right) \cdot
a_{n}^{\gamma_{1}-2}$.

\item
{Since $F$ and its first $\mu_j-1$ derivatives are equal to zero at $x = r_j$, by the Leibniz's rule for derivatives, we immediately know that for $i = 1,\ldots, s$ and $j$ satisfying $\mu_j\ge i$:
\begin{equation}\label{eqs:derivatives}
(F^{(i)}x^k)^{(\ell)}(r_j)=\left\{
\begin{array}{ll}
0,&\text{if}\ \ell<\mu_j-i;\\
F^{(\mu_j)}(r_j)\cdot r_j^k, &\text{if}\ \ell=\mu_j-i;\\
*,&\text{if}\ \ell>\mu_j-i.
\end{array}
\right.
\end{equation}
}

\item Plugging \eqref{eqs:derivatives} into \eqref{eq:rep_D_mu}, we have
\[
D(\boldsymbol{\gamma})=\frac{c\cdot{\footnotesize \left|
\setlength{\arraycolsep}{1.0pt}%
\begin{array}
[c]{lcll|c|lcll|l}%
~~~~~~~~0 & \cdots & ~~~~~~~~0 & F^{(\mu_{1})}(r_{1})r_{1}^{\gamma_{1}-1} &
\cdots & ~~~~~~~~0 & \cdots & ~~~~~~~~0 & \quad F^{(\mu_{i})}(r_{i}%
)r_{i}^{\gamma_{1}-1} & \cdots\\
~~~~~~~~\vdots &  & ~~~~~~~~\vdots & ~~~~~~~~\vdots &  & ~~~~~~~~\vdots &  &
~~~~~~~~\vdots & ~~~~~~~~\vdots & \\
~~~~~~~~0 & \cdots & ~~~~~~~~0 & F^{(\mu_{1})}(r_{1})r_{1}^{0} & \cdots &
~~~~~~~~0 & \cdots & ~~~~~~~~0 & \quad F^{(\mu_{i})}(r_{i})r_{i}^{0} &
\cdots\\\hline
~~~~~~~~0 & \cdots & F^{(\mu_{1})}(r_{1})r_{1}^{\gamma_{2}-1} &
~~~~~~~~\,\cdot & \cdots & ~~~~~~~~0 & \cdots & F^{(\mu_{i})}(r_{i}%
)r_{i}^{\gamma_{2}-1} & ~~~~~~~~\,\cdot & \cdots\\
~~~~~~~~\vdots &  & ~~~~~~~~\vdots & ~~~~~~~~\vdots &  & ~~~~~~~~\vdots &  &
~~~~~~~~\vdots & ~~~~~~~~\vdots & \\
~~~~~~~~0 & \cdots & F^{(\mu_{1})}(r_{1})r_{1}^{0} & ~~~~~~~~\,\cdot & \cdots
& ~~~~~~~~0 & \cdots & F^{(\mu_{i})}(r_{i})r_{i}^{0} & ~~~~~~~~\,\cdot &
\cdots\\\hline
~~~~~~~~\vdots &  & ~~~~~~~~\vdots & ~~~~~~~~\vdots &  & ~~~~~~~~\vdots &  &
~~~~~~~~\vdots & ~~~~~~~~\vdots & \\\hline
~~~~~~~~0 &  & ~~~~~~~~\,\cdot & ~~~~~~~~\cdot & \cdots & F^{(\mu_{i})}%
(r_{i})r_{i}^{\gamma_{\mu_{i}}-1} &  & ~~~~~~~~\,\cdot & ~~~~~~~~\,\cdot &
\cdots\\
~~~~~~~~\vdots &  & ~~~~~~~~\vdots & ~~~~~~~~\vdots &  & ~~~~~~~~\vdots &  &
~~~~~~~~\vdots & ~~~~~~~~\vdots & \\
~~~~~~~~0 &  & ~~~~~~~~\,\cdot & ~~~~~~~~\cdot & \cdots & F^{(\mu_{i})}%
(r_{i})r_{i}^{0} &  & ~~~~~~~~\,\cdot & ~~~~~~~~\,\cdot & \cdots\\\hline
~~~~~~~~\vdots &  & ~~~~~~~~\vdots & ~~~~~~~~\vdots &  & ~~~~~~~~\vdots &  &
~~~~~~~~\vdots & ~~~~~~~~\vdots & \\\hline
F^{(\mu_{1})}(r_{1})r_{1}^{\gamma_{\mu_{1}}-1} & \cdots & ~~~~~~~~\,\cdot &
~~~~~~~~\,\cdot & \cdots & ~~~~~~~~\,\cdot & \cdots & ~~~~~~~~\,\cdot &
~~~~~~~~\,\cdot & \cdots\\
~~~~~~~~\vdots &  & ~~~~~~~~\vdots & ~~~~~~~~\vdots &  & ~~~~~~~~\vdots &  &
~~~~~~~~\vdots & ~~~~~~~~\vdots & \\
F^{(\mu_{1})}(r_{1})r_{1}^{0} &  & ~~~~~~~~\,\cdot & ~~~~~~~~\,\cdot & \cdots
& ~~~~~~~~\,\cdot & \cdots & ~~~~~~~~\,\cdot & ~~~~~~~~\,\cdot & \cdots
\end{array}
\right| } }{\prod\limits_{1\leq i<j\leq m}(r_{i}-r_{j})^{\mu_{i}\mu_{j}}}%
\]
{where the block $D_{ij}$ at the $i$-th row and the $j$-th column is a matrix
$\setlength{\arraycolsep}{1pt}
\begin{pmatrix}
0&\cdots&0&F^{(\mu_j)}(r_j)r_j^{\gamma_i-1}&*&\cdots&*\\
\vdots&\ddots&\vdots&\vdots&\vdots&\ddots&\vdots&\\
0&\cdots&0&F^{(\mu_j)}(r_j)&*&\cdots&*
\end{pmatrix}
$
of size $\gamma_i\times\mu_j$.}

\item By rearranging the columns of the determinant in the numerator, we have
\begin{align*}
D(\boldsymbol{\gamma}) & =\pm\frac{c\cdot{\scriptsize \left|
\setlength{\arraycolsep}{1.0pt}
\begin{array}
[c]{lcl|ccc|lcl|lcl}
&  &  &  &  &  &  &  &  & F^{(\mu_{1})}(r_{1})r_{1}^{\gamma_{1}-1} & \cdots &
F^{(\mu_{m})}(r_{m})r_{m}^{\gamma_{1}-1}\\
&  &  &  &  &  &  &  &  & ~~~~~~~~\vdots &  & ~~~~~~~~\vdots\\
&  &  &  &  &  &  &  &  & F^{(\mu_{1})}(r_{1})r_{1}^{0} & \cdots & F^{(\mu
_{m})}(r_{m})r_{m}^{0}\\\hline
&  &  &  &  &  & F^{(\mu_{1})}(r_{1})r_{1}^{\gamma_{2}-1} & \cdots &
F^{(\mu_{\gamma_{2}})} (r_{\gamma_{2}})r_{\gamma_{2}}^{\gamma_{2}-1} &
~~~~~~~~\,\cdot & \cdots & ~~~~~~~~\,\cdot\\
&  &  &  &  &  & ~~~~~~~~\vdots &  & ~~~~~~~~\vdots & ~~~~~~~~\vdots &  &
~~~~~~~~\vdots\\
&  &  &  &  &  & F^{(\mu_{1})}(r_{1})r_{1}^{0} & \cdots & F^{(\mu_{\gamma_{2}%
})}(r_{\gamma_{2}})r_{\gamma_{2}}^{0} & ~~~~~~~~\,\cdot & \cdots &
~~~~~~~~\,\cdot\\\hline
&  &  & \cdot & \cdots & \cdot & ~~~~~~~~\,\cdot & \cdots & ~~~~~~~~\,\cdot &
~~~~~~~~\,\cdot & \cdots & ~~~~~~~~\,\cdot\\
&  &  & \vdots &  & \vdots & ~~~~~~~~\vdots &  & ~~~~~~~~\vdots &
~~~~~~~~\vdots &  & ~~~~~~~~\vdots\\
&  &  & \cdot & \cdots & \cdot & ~~~~~~~~\,\cdot & \cdots & ~~~~~~~~\,\cdot &
~~~~~~~~\,\cdot & \cdots & ~~~~~~~~\,\cdot\\\hline
F^{(\mu_{1})}(r_{1})r_{1}^{\gamma_{\mu_{1}}-1} & \cdots & F^{(\mu_{\gamma
_{\mu_{1}}})}(r_{\gamma_{\mu_{1}}})r_{\gamma_{\mu_{1}}}^{\gamma_{\mu_{1}}-1} &
\cdot & \cdots & \cdot & ~~~~~~~~\,\cdot & \cdots & ~~~~~~~~\,\cdot &
~~~~~~~~\,\cdot & \cdots & ~~~~~~~~\,\cdot\\
~~~~~~~~\vdots &  & ~~~~~~~~\vdots & \vdots &  & \vdots & ~~~~~~~~\vdots &  &
~~~~~~~~\vdots & ~~~~~~~~\vdots &  & ~~~~~~~~\vdots\\
F^{(\mu_{1})}(r_{1})r_{1}^{0} & \cdots & F^{(\mu_{\gamma_{\mu_{1}}}%
)}(r_{\gamma_{\mu_{1}}})r_{\gamma_{\mu_{1}}}^{0} & \cdot & \cdots & \cdot &
~~~~~~~~\,\cdot & \cdots & ~~~~~~~~\,\cdot & ~~~~~~~~\,\cdot & \cdots &
~~~~~~~~\,\cdot
\end{array}
\right| } } {\prod\limits_{1\leq i<j\leq m}(r_{i}-r_{j})^{\mu_{i}\mu_{j}}}\\
& =\pm\frac{c\cdot\left|
\begin{array}
[c]{cccc}
&  &  & M_{1}\\
&  & M_{2} & \cdot\\[3pt]
& \!\!\!\!\!\!\!\begin{sideways}$\ddots$\end{sideways} &  & \cdot\\
M_{\mu_{1}} & \cdots & \cdot & \cdot
\end{array}
\right| }{\prod\limits_{1\leq i<j\leq m}(r_{i}-r_{j})^{\mu_{i}\mu_{j}}}%
\end{align*}
where
\[
M_{i}= \left[
\begin{array}
[c]{lcl}%
F^{(\mu_{1})}(r_{1})r_{1}^{\gamma_{i}-1} & \cdots & F^{(\mu_{\gamma_{i}}%
)}(r_{\gamma_{i}})r_{\gamma_{i}}^{\gamma_{i}-1}\\
~~~~~~~~\vdots &  & ~~~~~~~~\vdots\\
F^{(\mu_{1})}(r_{1})r_{1}^{0} & \cdots & F^{(\mu_{\gamma_{i}})}(r_{\gamma_{i}%
})r_{\gamma_{i}}^{0}%
\end{array}
\right]
\]
for $i=1,\ldots,\mu_{1}$. Then
\[
D(\boldsymbol{\gamma})=\pm\frac{c\cdot\left|  M_{1}\right| \cdots\left|
M_{\mu_{1}}\right| }{\prod\limits_{1\leq i<j\leq m}(r_{i}-r_{j})^{\mu_{i}%
\mu_{j}}}.%
\]

\item It only remains to show that $\left\vert M_{i}\right\vert\ne0$. The claim follows from the
following observations:
\[
\left|  M_{i}\right| =\left( \prod_{j=1}^{\gamma_{i}}F^{(\mu_{j})}%
(r_{j})\right) V(r_{1},\ldots,r_{\gamma_{i}})\ne0 \ \ \text{for}%
\ \ i=1,\ldots,\mu_{1}.%
\]
The proof is completed.
\end{enumerate}
\end{proof}

\begin{lemma}
\label{lem:D_lambda_all_zeros} Let $\operatorname*{mult}(F)=\boldsymbol{\mu}$.
Then $D(\boldsymbol{\lambda})=0$ for any $\boldsymbol{\lambda}$ such that
$\boldsymbol{\bar{\mu}}\prec_{\operatorname*{lex}}\boldsymbol{\lambda}$.
\end{lemma}

\begin{proof}
In order to convey the main underlying ideas, we will show the proof for a
particular case first. After that, we will generalize the ideas to arbitrary cases.

\bigskip

\noindent\textbf{Particular case:} Consider the case $n=5$ and
$\operatorname*{mult}(F)=\boldsymbol{\mu}=\left(  3,2\right) $. Let
$\boldsymbol{\gamma}=\bar{\boldsymbol{\mu}}=(2,2,1)$ and $\boldsymbol{\lambda
}=(3,1,1)$. Obviously, $\bar{\boldsymbol{\mu}}\prec_{\operatorname*{lex}%
}\boldsymbol{\lambda}$. We will show that $D(\boldsymbol{\lambda})=0$.

\begin{enumerate}
\item Assume that $r_{1}$ and $r_{2}$ are the two distinct roots with
multiplicities $3$ and $2,$ respectively. In other words, $F=a_{5}%
(x-r_{1})^{3}(x-r_{2})^{2}$.

\item By Lemma \ref{lem:D_mu_in_mroots},
\begin{align*}
D(\boldsymbol{\lambda})  & =\frac{c\cdot\left|
\setlength{\arraycolsep}{8.0pt}
\begin{array}
[c]{lcl|cc}%
({F^{(1)}x^{2}})^{(0)}(r_{1}) & ({F^{(1)}x^{2}})^{(1)}(r_{1}) & ({F^{(1)}%
x^{1}})^{(2)}(r_{1}) & ({F^{(1)}x^{2}})^{(0)}(r_{2}) & ({F^{(1)}x^{2}}%
)^{(1)}(r_{2})\\
({F^{(1)}x^{1}})^{(0)}(r_{1}) & ({F^{(1)}x^{1}})^{(1)}(r_{1}) & ({F^{(1)}%
x^{1}})^{(2)}(r_{1}) & ({F^{(1)}x^{1}})^{(0)}(r_{2}) & ({F^{(1)}x^{1}}%
)^{(1)}(r_{2})\\
({F^{(1)}x^{0}})^{(0)}(r_{1}) & ({F^{(1)}x^{0}})^{(1)}(r_{1}) & ({F^{(1)}%
x^{0}})^{(2)}(r_{1}) & ({F^{(1)}x^{0}})^{(0)}(r_{2}) & ({F^{(1)}x^{0}}%
)^{(1)}(r_{2})\\\hline
({F^{(2)}x^{0}})^{(0)}(r_{1}) & ({F^{(2)}x^{0}})^{(1)}(r_{1}) & ({F^{(2)}%
x^{0}})^{(2)}(r_{1}) & ({F^{(2)}x^{0}})^{(0)}(r_{2}) & ({F^{(2)}x^{0}}%
)^{(1)}(r_{2})\\\hline
({F^{(3)}x^{0}})^{(0)}(r_{1}) & ({F^{(3)}x^{0}})^{(1)}(r_{1}) & ({F^{(3)}%
x^{0}})^{(2)}(r_{1}) & ({F^{(3)}x^{0}})^{(0)}(r_{2}) & ({F^{(3)}x^{0}}%
)^{(1)}(r_{2})
\end{array}
\right|  }{(r_{1}-r_{2})^{6}}%
\end{align*}
where
\[
c=\pm1/\big[(0!\cdot1!\cdot2!)\cdot(0!\cdot1!)]\cdot a_{5}^{1}=\pm a_{5}/2.
\]

\item Recall \eqref{eqs:derivatives_of_ri}. Then we immediately have {\small
\[
\setlength{\arraycolsep}{1.0pt}
\begin{array}
[c]{l}%
({F^{(1)}x^{2}})^{(0)}(r_{1})=0,\ ({F^{(1)}x^{2}})^{(1)}(r_{1}%
)=0,\ ({F^{(1)}x^{2}})^{(2)}(r_{1})=F^{(3)}(r_{1})r_{1}^{2},\ ({F^{(1)}%
x^{2}})^{(0)}(r_{2})=0,\ ({F^{(1)}x^{2}})^{(1)}(r_{2})=F^{(2)}(r_{2}%
)r_{2}^{2},\\[2pt]%
({F^{(1)}x^{1}})^{(0)}(r_{1})=0,\ ({F^{(1)}x^{1}})^{(1)}(r_{1}%
)=0,\ ({F^{(1)}x^{1}})^{(2)}(r_{1})=F^{(3)}(r_{1})r_{1}^{1},\ ({F^{(1)}%
x^{1}})^{(0)}(r_{2})=0,\ ({F^{(1)}x^{1}})^{(1)}(r_{2})=F^{(2)}(r_{2}%
)r_{2}^{1},\\[2pt]%
({F^{(1)}x^{0}})^{(0)}(r_{1})=0,\ ({F^{(1)}x^{0}})^{(1)}(r_{1}%
)=0,\ ({F^{(1)}x^{0}})^{(2)}(r_{1})=F^{(3)}(r_{1})r_{1}^{0},\ ({F^{(1)}%
x^{0}})^{(0)}(r_{2})=0,\ ({F^{(1)}x^{0}})^{(1)}(r_{2})=F^{(2)}(r_{2}%
)r_{2}^{0},\\[2pt]%
({F^{(2)}x^{0}})^{(0)}(r_{1})=0,\ ({F^{(2)}x^{0}})^{(1)}(r_{1})=F^{(3)}%
(r_{1})r_{1}^{0}, \hspace{9.5em} ({F^{(2)}x^{0}})^{(0)}(r_{2})={F^{(2)}%
(r_{2})r_{2}^{0}},\\[2pt]%
({F^{(3)}x^{0}})^{(0)}(r_{1})=F^{(3)}(r_{1})r_{1}^{0}.%
\end{array}
\]
}

Therefore,
\[
D(\boldsymbol{\lambda}) =\frac{c\cdot\left|
\begin{array}
[c]{ccc|cc}%
0 & 0 & F^{(3)}(r_{1})r_{1}^{2} & 0 & F^{(2)}(r_{2})r_{2}^{2}\\
0 & 0 & F^{(3)}(r_{1})r_{1}^{1} & 0 & F^{(2)}(r_{2})r_{2}^{1}\\
0 & 0 & F^{(3)}(r_{1})r_{1}^{0} & 0 & F^{(2)}(r_{2})r_{2}^{0}\\\hline
0 & F^{(3)}(r_{1})r_{1}^{0} & \cdot & F^{(2)}(r_{2})r_{2}^{0} & \cdot\\\hline
F^{(3)}(r_{1})r_{1}^{0} & \cdot & \cdot & \cdot & \cdot
\end{array}
\right|  }{(r_{1}-r_{2})^{6}}.%
\]

\item By rearranging the columns of the determinant in the numerator, we have
\begin{align*}
D(\boldsymbol{\lambda}) &  =\pm\frac{c\cdot\left|
\begin{array}
[c]{c|cc|cc}
&  &  & F^{(3)}(r_{1})r_{1}^{2} & F^{(2)}(r_{2})r_{2}^{2}\\
&  &  & F^{(3)}(r_{1})r_{1}^{1} & F^{(2)}(r_{2})r_{2}^{1}\\
&  &  & F^{(3)}(r_{1})r_{1}^{0} & F^{(2)}(r_{2})r_{2}^{0}\\\hline
& F^{(3)}(r_{1})r_{1}^{0} & F^{(2)}(r_{2})r_{2}^{0} & \cdot & \cdot\\\hline
F^{(3)}(r_{1})r_{1}^{0} & \cdot & \cdot & \cdot & \cdot
\end{array}
\right|  }{(r_{1}-r_{2})^{6}}\\
& =\pm\frac{c\cdot\left|
\begin{array}
[c]{ccc}
&  & M_{1}\\
& M_{2} & \cdot\\
M_{3} & \cdot & \cdot
\end{array}
\right| }{(r_{1}-r_{2})^{6}}%
\end{align*}
where
\[
M_{1}= \left[
\begin{array}
[c]{ll}%
F^{(3)}(r_{1})r_{1}^{2} & F^{(2)}(r_{2})r_{2}^{2}\\
F^{(3)}(r_{1})r_{1}^{1} & F^{(2)}(r_{2})r_{2}^{1}\\
F^{(3)}(r_{1})r_{1}^{0} & F^{(2)}(r_{2})r_{2}^{0}%
\end{array}
\right],  \quad M_{2}= \left[
\begin{array}
[c]{ll}%
F^{(3)}(r_{1})r_{1}^{0} & F^{(2)}(r_{2})r_{2}^{0}%
\end{array}
\right],  \quad M_{3}= \left[
\begin{array}
[c]{l}%
F^{(3)}(r_{1})r_{1}^{0}%
\end{array}
\right].
\]

\item We repartition the columns so that the reverse diagonal consists of two
square matrices and obtain the following:
\begin{align*}
D(\boldsymbol{\lambda}) &  =\pm\frac{c\cdot\left|
\begin{array}
[c]{cc|ccc}
&  &  & F^{(3)}(r_{1})r_{1}^{2} & F^{(2)}(r_{2})r_{2}^{2}\\
&  &  & F^{(3)}(r_{1})r_{1}^{1} & F^{(2)}(r_{2})r_{2}^{1}\\
&  &  & F^{(3)}(r_{1})r_{1}^{0} & F^{(2)}(r_{2})r_{2}^{0}\\\hline
& F^{(3)}(r_{1})r_{1}^{0} & F^{(2)}(r_{2})r_{2}^{0} & \cdot & \cdot\\
F^{(3)}(r_{1})r_{1}^{0} & \cdot & \cdot & \cdot & \cdot
\end{array}
\right|  }{(r_{1}-r_{2})^{6}}\\
& =\pm\frac{c\cdot\left|
\begin{array}
[c]{cc}
& T\\
B & \cdot
\end{array}
\right|  }{(r_{1}-r_{2})^{6}}%
\end{align*}
where the size of the square matrix $T$ is $\lambda_{1}=3$, namely,
\[
T=\left[
\begin{array}
[c]{cc}%
0 & M_{1}%
\end{array}
\right],
\]
where $0$ is the $\lambda_{1}\times(\lambda_{1}-\gamma_{1})$ matrix.

\item Since $\lambda_{1}-\gamma_{1}=3-2>0$, the first column of $T$ is all
zeros. Hence $|T|=0$ and in turn $D(\boldsymbol{\lambda}) =0$.
\end{enumerate}

\medskip

\noindent\textbf{Arbitrary case.} Now we generalize the above ideas to
arbitrary cases.

\begin{enumerate}
\item Let $\boldsymbol{\mu}=(\mu_{1},\ldots,\mu_{m})$. Assume that
$r_{1},\ldots,r_{m}$ are the $m$ distinct roots of $F$ with multiplicities
$\mu_{1},\ldots,\mu_{m} $ respectively. In other words, $F=a_{n}(x-r_{1}%
)^{\mu_{1}}\cdots(x-r_{m})^{\mu_{m}}$.

\item Let $\boldsymbol{\gamma}=\bar{\boldsymbol{\mu}}=(\gamma_{1}%
,\ldots,\gamma_{s})$. By the definition of conjugate, $\gamma_{i}=\#\{\mu
_{j}:\,\mu_{j}\ge i\}$. Note that $\gamma_{1}=m$ and $s=\mu_{1}$ since
$\mu_{1}\ge\cdots\ge\mu_{m}\ge1$.

\item Consider $\boldsymbol{\lambda}=(\lambda_{1},\ldots,\lambda_{t}%
)\in\mathcal{M}(n)$ such that $\boldsymbol{\gamma}\prec_{\text{lex}%
}\boldsymbol{\lambda}$. By Lemma \ref{lem:D_mu_in_mroots}, we have
\begin{equation}
\label{eq:D_lambda}D(\boldsymbol{\lambda})=\frac{c\cdot{\footnotesize \left|
\setlength{\arraycolsep}{1.0pt}%
\begin{array}
[c]{lcl|c|lcl}%
({F^{(1)}x^{\lambda_{1}-1}})^{(0)}(r_{1}) & \cdots & ({F^{(1)}x^{\lambda
_{1}-1}})^{(\mu_{1}-1)}(r_{1}) & \cdots\cdots & ({F^{(1)}x^{\lambda_{1}-1}%
})^{(0)}(r_{m}) & \cdots & ({F^{(1)}x^{\lambda_{1}-1}})^{(\mu_{m}-1)}(r_{m})\\[-2pt]
~~~~~~~~\vdots &  & ~~~~~~~~\vdots &  & ~~~~~~~~\vdots &  & ~~~~~~~~\vdots\\[-2pt]
({F^{(1)}x^{0}})^{(0)}(r_{1}) & \cdots & ({F^{(1)}x^{0}})^{(\mu_{1}-1)}%
(r_{1}) & \cdots\cdots & ({F^{(1)}x^{0}})^{(0)}(r_{m}) & \cdots &
({F^{(1)}x^{0}})^{(\mu_{m}-1)}(r_{m})\\[-2pt]\hline
~~~~~~~~\vdots &  & ~~~~~~~~\vdots &  & ~~~~~~~~\vdots &  & ~~~~~~~~\vdots
\\[-2pt]\hline
({F^{(t)}}x^{\lambda_{t}-1})^{(0)}(r_{1}) & \cdots & ({F}^{(t)}x^{\lambda
_{t}-1})^{(\mu_{1}-1)}(r_{1}) & \cdots\cdots & ({F^{(t)}}x^{\lambda_{t}%
-1})^{(0)}(r_{m}) & \cdots & ({F}^{(t)}x^{\lambda_{t}-1})^{(\mu_{m}-1)}%
(r_{m})\\[-2pt]
~~~~~~~~\vdots &  & ~~~~~~~~\vdots &  & ~~~~~~~~\vdots &  & ~~~~~~~~\vdots\\[-2pt]
({F^{(t)}}x^{0})^{(0)}(r_{1}) & \cdots & ({F^{(t)}}x^{0})^{(\mu_{1}-1)}%
(r_{1}) & \cdots\cdots & ({F^{(t)}}x^{0})^{(0)}(r_{m}) & \cdots & ({F^{(t)}%
}x^{0})^{(\mu_{m}-1)}(r_{m})
\end{array}
\right|}  }{\prod_{i<j}(r_{i}-r_{j})^{\mu_{i}\mu_{j}}}%
\end{equation}
where $c=\pm1\Big/\left( \prod_{i=1}^{m}\prod_{j=0}^{\mu_{i}-1}j!\right) \cdot
a_{n}^{\lambda_{1}-2}$.

\item Recall \eqref{eqs:derivatives} and plug \eqref{eqs:derivatives} into
\eqref{eq:D_lambda}. Then we get
\[
D(\boldsymbol{\gamma})=\frac{c\cdot{\scriptsize \left|
\setlength{\arraycolsep}{1.0pt}%
\begin{array}
[c]{lcll|c|lcll|l}%
~~~~~~~~0 & \cdots & ~~~~~~~~0 & F^{(\mu_{1})}(r_{1})r_{1}^{\lambda_{1}-1} &
\cdots & ~~~~~~~~0 & \cdots & ~~~~~~~~0 & \quad F^{(\mu_{i})}(r_{i}%
)r_{i}^{\lambda_{1}-1} & \cdots\\[-2pt]
~~~~~~~~\vdots &  & ~~~~~~~~\vdots & ~~~~~~~~\vdots &  & ~~~~~~~~\vdots &  &
~~~~~~~~\vdots & ~~~~~~~~\vdots & \\[-2pt]
~~~~~~~~0 & \cdots & ~~~~~~~~0 & F^{(\mu_{1})}(r_{1})r_{1}^{0} & \cdots &
~~~~~~~~0 & \cdots & ~~~~~~~~0 & \quad F^{(\mu_{i})}(r_{i})r_{i}^{0} &
\cdots\\[-2pt]\hline
~~~~~~~~0 & \cdots & F^{(\mu_{1})}(r_{1})r_{1}^{\lambda_{2}-1} &
~~~~~~~~\,\cdot & \cdots & ~~~~~~~~0 & \cdots & F^{(\mu_{i})}(r_{i}%
)r_{i}^{\lambda_{2}-1} & ~~~~~~~~\,\cdot & \cdots\\[-2pt]
~~~~~~~~\vdots &  & ~~~~~~~~\vdots & ~~~~~~~~\vdots &  & ~~~~~~~~\vdots &  &
~~~~~~~~\vdots & ~~~~~~~~\vdots & \\[-2pt]
~~~~~~~~0 & \cdots & F^{(\mu_{1})}(r_{1})r_{1}^{0} & ~~~~~~~~\,\cdot & \cdots
& ~~~~~~~~0 & \cdots & F^{(\mu_{i})}(r_{i})r_{i}^{0} & ~~~~~~~~\,\cdot &
\cdots\\[-2pt]\hline
~~~~~~~~\vdots &  & ~~~~~~~~\vdots & ~~~~~~~~\vdots &  & ~~~~~~~~\vdots &  &
~~~~~~~~\vdots & ~~~~~~~~\vdots & \\[-2pt]\hline
~~~~~~~~0 &  & ~~~~~~~~\,\cdot & ~~~~~~~~\cdot & \cdots & F^{(\mu_{i})}%
(r_{i})r_{i}^{\lambda_{\mu_{i}}-1} &  & ~~~~~~~~\,\cdot & ~~~~~~~~\,\cdot &
\cdots\\[-2pt]
~~~~~~~~\vdots &  & ~~~~~~~~\vdots & ~~~~~~~~\vdots &  & ~~~~~~~~\vdots &  &
~~~~~~~~\vdots & ~~~~~~~~\vdots & \\[-2pt]
~~~~~~~~0 &  & ~~~~~~~~\,\cdot & ~~~~~~~~\cdot & \cdots & F^{(\mu_{i})}%
(r_{i})r_{i}^{0} &  & ~~~~~~~~\,\cdot & ~~~~~~~~\,\cdot & \cdots\\\hline
~~~~~~~~\vdots &  & ~~~~~~~~\vdots & ~~~~~~~~\vdots &  & ~~~~~~~~\vdots &  &
~~~~~~~~\vdots & ~~~~~~~~\vdots & \\[-2pt]\hline
F^{(\mu_{1})}(r_{1})r_{1}^{\lambda_{\mu_{1}}-1} & \cdots & ~~~~~~~~\,\cdot &
~~~~~~~~\,\cdot & \cdots & ~~~~~~~~\,\cdot & \cdots & ~~~~~~~~\,\cdot &
~~~~~~~~\,\cdot & \cdots\\[-2pt]
~~~~~~~~\vdots &  & ~~~~~~~~\vdots & ~~~~~~~~\vdots &  & ~~~~~~~~\vdots &  &
~~~~~~~~\vdots & ~~~~~~~~\vdots & \\[-2pt]
F^{(\mu_{1})}(r_{1})r_{1}^{0} &  & ~~~~~~~~\,\cdot & ~~~~~~~~\,\cdot & \cdots
& ~~~~~~~~\,\cdot & \cdots & ~~~~~~~~\,\cdot & ~~~~~~~~\,\cdot & \cdots
\end{array}
\right| } }{\prod\limits_{1\leq i<j\leq m}(r_{i}-r_{j})^{\mu_{i}\mu_{j}}}.%
\]

\item By rearranging the columns of the determinant in the numerator, we have
\begin{align*}
D(\boldsymbol{\lambda}) & =\pm\frac{c\cdot{\scriptsize \left|
\setlength{\arraycolsep}{1.0pt}
\begin{array}
[c]{lcl|ccc|lcl|lcl}
&  &  &  &  &  &  &  &  & F^{(\mu_{1})}(r_{1})r_{1}^{\lambda_{1}-1} & \cdots &
F^{(\mu_{\gamma_{1}})}(r_{\gamma_{1}})r_{\gamma_{1}}^{\lambda_{1}-1}\\
&  &  &  &  &  &  &  &  & ~~~~~~~~\vdots &  & ~~~~~~~~\vdots\\
&  &  &  &  &  &  &  &  & F^{(\mu_{1})}(r_{1})r_{1}^{0} & \cdots &
F^{(\mu_{\gamma_{1}})}(r_{\gamma_{1}})r_{\gamma_{1}}^{0}\\\hline
&  &  &  &  &  & F^{(\mu_{1})}(r_{1})r_{1}^{\lambda_{2}-1} & \cdots &
F^{(\mu_{\gamma_{2}})} (r_{\gamma_{2}})r_{\gamma_{2}}^{\lambda_{2}-1} &
~~~~~~~~\,\cdot & \cdots & ~~~~~~~~\,\cdot\\
&  &  &  &  &  & ~~~~~~~~\vdots &  & ~~~~~~~~\vdots & ~~~~~~~~\vdots &  &
~~~~~~~~\vdots\\
&  &  &  &  &  & F^{(\mu_{1})}(r_{1})r_{1}^{0} & \cdots & F^{(\mu_{\gamma_{2}%
})}(r_{\gamma_{2}})r_{\gamma_{2}}^{0} & ~~~~~~~~\,\cdot & \cdots &
~~~~~~~~\,\cdot\\\hline
&  &  & \cdot & \cdots & \cdot & ~~~~~~~~\,\cdot & \cdots & ~~~~~~~~\,\cdot &
~~~~~~~~\,\cdot & \cdots & ~~~~~~~~\,\cdot\\
&  &  & \vdots &  & \vdots & ~~~~~~~~\vdots &  & ~~~~~~~~\vdots &
~~~~~~~~\vdots &  & ~~~~~~~~\vdots\\
&  &  & \cdot & \cdots & \cdot & ~~~~~~~~\,\cdot & \cdots & ~~~~~~~~\,\cdot &
~~~~~~~~\,\cdot & \cdots & ~~~~~~~~\,\cdot\\\hline
F^{(\mu_{1})}(r_{1})r_{1}^{\lambda_{t}-1} & \cdots & F^{(\mu_{\gamma_{s}}%
)}(r_{\gamma_{s}})r_{\gamma_{s}}^{\lambda_{t}-1} & \cdot & \cdots & \cdot &
~~~~~~~~\,\cdot & \cdots & ~~~~~~~~\,\cdot & ~~~~~~~~\,\cdot & \cdots &
~~~~~~~~\,\cdot\\
~~~~~~~~\vdots &  & ~~~~~~~~\vdots & \vdots &  & \vdots & ~~~~~~~~\vdots &  &
~~~~~~~~\vdots & ~~~~~~~~\vdots &  & ~~~~~~~~\vdots\\
F^{(\mu_{1})}(r_{1})r_{1}^{0} & \cdots & F^{(\mu_{\gamma_{s}})}(r_{\gamma_{s}%
})r_{\gamma_{s}}^{0} & \cdot & \cdots & \cdot & ~~~~~~~~\,\cdot & \cdots &
~~~~~~~~\,\cdot & ~~~~~~~~\,\cdot & \cdots & ~~~~~~~~\,\cdot
\end{array}
\right| } } {\prod\limits_{1\leq i<j\leq m}(r_{i}-r_{j})^{\mu_{i}\mu_{j}}}\\
& =\pm\frac{c\cdot\left|
\begin{array}
[c]{cccc}
&  &  & M_{1}\\
&  & M_{2} & \cdot\\[3pt]
& \!\!\!\!\!\!\!\begin{sideways}$\ddots$\end{sideways} &  & \cdot\\
M_{\mu_{1}} & \cdots & \cdot & \cdot
\end{array}
\right| }{\prod\limits_{1\leq i<j\leq m}(r_{i}-r_{j})^{\mu_{i}\mu_{j}}}%
\end{align*}
where $M_{i}$ is $\lambda_{i}$ by $\gamma_{i}$.

\item Since $\boldsymbol{\gamma}\prec_{\text{lex}}\boldsymbol{\lambda}$, there
exists $\ell$ such that $\gamma_{j}=\lambda_{j}$ for $j<\ell$ and
$\gamma_{\ell}<\lambda_{\ell}$. Thus
\[
\gamma_{1}+\cdots+\gamma_{\ell}<\lambda_{1}+\cdots+\lambda_{\ell}.%
\]

\item We repartition the numerator matrix so that the reverse diagonal
consists of two square matrices $T$ and $B$ as follows:
\[
D(\boldsymbol{\lambda})=\pm\frac{c\cdot\left|
\begin{array}
[c]{cc}
& T\\
B &{\cdot}
\end{array}
\right| }{\prod\limits_{1\leq i<j\leq m}(r_{i}-r_{j})^{\mu_{i}\mu_{j}}}%
\]
where the size of the square matrix $T$ is $\lambda_{1}+\cdots+\lambda_{\ell}%
$, namely,
\[
T=\left[
\begin{array}
[c]{c|ccc}
&  &  & M_{1}\\
0~ &  & \begin{sideways}$\ddots$\end{sideways} & \vdots\\
& M_{\ell} & \cdots & \cdot
\end{array}
\right],
\]
where again $0$ is the $\gamma_{\ell}\times p$ and $p=(\lambda_{1}%
+\cdots+\lambda_{\ell})-(\gamma_{1}+\cdots+\gamma_{\ell})$.

\item Obviously,
\[
D(\boldsymbol{\lambda})=\pm\frac{c\cdot\left|  T\right| \cdot\left|  B\right|
}{\prod\limits_{1\leq i<j\leq m}(r_{i}-r_{j})^{\mu_{i}\mu_{j}}}.%
\]

\item Since $p>0$, the first column of $T$ is all zeros. Hence $\left|
T\right| =0$, which implies that $D(\boldsymbol{\lambda})=0$.
\end{enumerate}
\end{proof}

\subsection{Proof of Theorem \ref{thm:main_result}}

Now we are ready to prove Theorem \ref{thm:main_result}.

\begin{proof}
[Proof of Theorem \ref{thm:main_result}]\ Theorem \ref{thm:main_result} is
equivalent to the following claim: let%
\[
\boldsymbol{\delta}=\max\limits_{\substack{\boldsymbol{\gamma}\in
\mathcal{M}(n) \\D(\boldsymbol{\gamma})\ne0}}\boldsymbol{\gamma}%
\]
where $\max$ is with respect to the lexicographic ordering $\prec
_{\operatorname*{lex}}$. Then $\operatorname*{mult}(F)=\overline
{\boldsymbol{\delta}}$.

Next we will show the correctness of the claim.

\begin{enumerate}
\item Assume that $\text{mult}(F)=\boldsymbol{\mu}$. We will show
$\boldsymbol{\mu}=\overline{\boldsymbol{\delta}}$ by disproving
{$\boldsymbol{\delta}\prec_{\text{lex}}\overline{\boldsymbol{\mu}}$}
and {$\overline{\boldsymbol{\mu}}~\prec_{\text{lex}}\boldsymbol{\delta}$}.

\item If
{$\boldsymbol{\delta}\prec_{\text{lex}}\overline{\boldsymbol{\mu}}$, then by}
the condition for determining $\boldsymbol{\delta}$, we immediately have
$D(\overline{\boldsymbol{\mu}}) =0$, leading to a contradiction with Lemma
\ref{lem:D_mu_not_zero}.

\item
If {$\overline{\boldsymbol{\mu}}\prec_{\text{lex}}\boldsymbol{\delta} $,
then by} Lemma \ref{lem:D_lambda_all_zeros}, $D(\boldsymbol{\delta}) =0.$
However, it contradicts the condition for determining $\boldsymbol{\delta}$.

\item Therefore, the only possibility is $\boldsymbol{\mu}=\overline
{\boldsymbol{\delta}}$.
\end{enumerate}
\end{proof}

\section{Comparison}

\label{sec:comparison}

In this section, we compare the multiplicity discriminant condition given by
Theorem \ref{thm:main_result} (mentioned as HY22 hereinafter) and that given
by a complex root version of YHZ's condition \cite{1996_Yang_Hou_Zeng} as well
as the one given by the authors in \cite[Theorem 6]{2021_Hong_Yang} (mentioned
as HY21 hereinafter). In particular, we will make comparison on the forms and
the maximum degrees of discriminants appearing in the conditions.

\subsection{Form of discriminants}

We will illustrate the forms of conditions generated by the three methods for
a fixed $\boldsymbol{\mu}$. For example, we consider the polynomial
$F=a_{5}x^{5}+a_{4}x^{4}+a_{3}x^{3}+a_{2}x^{2}+a_{1}x+a_{0}$ and
$\boldsymbol{\mu}=(2,2,1)$. The condition for $F$ having the multiplicity
structure $\boldsymbol{\mu}$ is given as follows:

\begin{enumerate}
\item YHZ's condition: $P_{1}=0\wedge P_{2}=0\wedge P_{3}\neq0$ where
\begin{align*}
P_{1}= &  \left\vert
\begin{array}
[c]{ccccccccc}%
a_{5} & a_{4} & a_{3} & a_{2} & a_{1} & a_{0} &  &  & \\[-2pt]
& a_{5} & a_{4} & a_{3} & a_{2} & a_{1} & a_{0} &  & \\[-2pt]
&  & a_{5} & a_{4} & a_{3} & a_{2} & a_{1} & a_{0} & \\[-2pt]
&  &  & a_{5} & a_{4} & a_{3} & a_{2} & a_{1} & a_{0}\\[-2pt]
5a_{5} & 4a_{4} & 3a_{3} & 2a_{2} & a_{1} &  &  &  & \\[-2pt]
& 5a_{5} & 4a_{4} & 3a_{3} & 2a_{2} & a_{1} &  &  & \\[-2pt]
&  & 5a_{5} & 4a_{4} & 3a_{3} & 2a_{2} & a_{1} &  & \\[-2pt]
&  &  & 5a_{5} & 4a_{4} & 3a_{3} & 2a_{2} & a_{1} & \\[-2pt]
&  &  &  & 5a_{5} & 4a_{4} & 3a_{3} & 2a_{2} & a_{1}%
\end{array}
\right\vert, \\[5pt]
P_{2}= &  \left\vert
\begin{array}
[c]{ccccccc}%
a_{5} & a_{4} & a_{3} & a_{2} & a_{1} & a_{0} & \\[-2pt]
& a_{5} & a_{4} & a_{3} & a_{2} & a_{1} & a_{0}\\[-2pt]
&  & a_{5} & a_{4} & a_{3} & a_{2} & a_{1}\\[-2pt]
5a_{5} & 4a_{4} & 3a_{3} & 2a_{2} & a_{1} &  & \\[-2pt]
& 5a_{5} & 4a_{4} & 3a_{3} & 2a_{2} & a_{1} & \\[-2pt]
&  & 5a_{5} & 4a_{4} & 3a_{3} & 2a_{2} & a_{1}\\[-2pt]
&  &  & 5a_{5} & 4a_{4} & 3a_{3} & 2a_{2}%
\end{array}
\right\vert, \\[5pt]
P_{3}= &  \left\vert
\begin{array}
[c]{rrr}%
\left\vert
\begin{array}
[c]{ccccc}%
a_{5} & a_{4} & a_{3} & a_{2} & a_{1}\\[-2pt]
& a_{5} & a_{4} & a_{3} & a_{2}\\[-2pt]
5a_{5} & 4a_{4} & 3a_{3} & 2a_{2} & a_{1}\\[-2pt]
& 5a_{5} & 4a_{4} & 3a_{3} & 2a_{2}\\[-2pt]
&  & 5a_{5} & 4a_{4} & 3a_{3}%
\end{array}
\right\vert  & \left\vert
\begin{array}
[c]{ccccc}%
a_{5} & a_{4} & a_{3} & a_{2} & a_{0}\\[-2pt]
& a_{5} & a_{4} & a_{3} & a_{1}\\[-2pt]
5a_{5} & 4a_{4} & 3a_{3} & 2a_{2} & \\[-2pt]
& 5a_{5} & 4a_{4} & 3a_{3} & a_{1}\\[-2pt]
&  & 5a_{5} & 4a_{4} & 2a_{2}%
\end{array}
\right\vert  & \left\vert
\begin{array}
[c]{ccccc}%
a_{5} & a_{4} & a_{3} & a_{2} & \\[-2pt]
& a_{5} & a_{4} & a_{3} & a_{0}\\[-2pt]
5a_{5} & 4a_{4} & 3a_{3} & 2a_{2} & \\[-2pt]
& 5a_{5} & 4a_{4} & 3a_{3} & \\[-2pt]
&  & 5a_{5} & 4a_{4} & a_{1}%
\end{array}
\right\vert \\[35pt]%
2\left\vert
\begin{array}
[c]{ccccc}%
a_{5} & a_{4} & a_{3} & a_{2} & a_{1}\\[-2pt]
& a_{5} & a_{4} & a_{3} & a_{2}\\[-2pt]
5a_{5} & 4a_{4} & 3a_{3} & 2a_{2} & a_{1}\\[-2pt]
& 5a_{5} & 4a_{4} & 3a_{3} & 2a_{2}\\[-2pt]
&  & 5a_{5} & 4a_{4} & 3a_{3}%
\end{array}
\right\vert  & \left\vert
\begin{array}
[c]{ccccc}%
a_{5} & a_{4} & a_{3} & a_{2} & a_{0}\\[-2pt]
& a_{5} & a_{4} & a_{3} & a_{1}\\[-2pt]
5a_{5} & 4a_{4} & 3a_{3} & 2a_{2} & \\[-2pt]
& 5a_{5} & 4a_{4} & 3a_{3} & a_{1}\\[-2pt]
&  & 5a_{5} & 4a_{4} & 2a_{2}%
\end{array}
\right\vert  & \\[35pt]
& 2\left\vert
\begin{array}
[c]{ccccc}%
a_{5} & a_{4} & a_{3} & a_{2} & a_{1}\\[-2pt]
& a_{5} & a_{4} & a_{3} & a_{2}\\[-2pt]
5a_{5} & 4a_{4} & 3a_{3} & 2a_{2} & a_{1}\\[-2pt]
& 5a_{5} & 4a_{4} & 3a_{3} & 2a_{2}\\[-2pt]
&  & 5a_{5} & 4a_{4} & 3a_{3}%
\end{array}
\right\vert
& \left\vert
\begin{array}
[c]{ccccc}%
a_{5} & a_{4} & a_{3} & a_{2} & a_{0}\\[-2pt]
& a_{5} & a_{4} & a_{3} & a_{1}\\[-2pt]
5a_{5} & 4a_{4} & 3a_{3} & 2a_{2} & \\[-2pt]
& 5a_{5} & 4a_{4} & 3a_{3} & a_{1}\\[-2pt]
&  & 5a_{5} & 4a_{4} & 2a_{2}%
\end{array}
\right\vert
\end{array}
\right\vert.
\end{align*}

\item HY21's condition: ${Q_{1}=0\wedge Q_{2}=0}\wedge Q_{3}\neq0\wedge
Q_{4}\neq0$ where $Q_1=P_1$, $Q_2=P_2$ and
\begin{align*}
\hspace{3.2em} Q_{3}= &  \left\vert
\begin{array}
[c]{ccccc}%
a_{5} & a_{4} & a_{3} & a_{2} & a_{1}\\[-2pt]
& a_{5} & a_{4} & a_{3} & a_{2}\\[-2pt]
5a_{5} & 4a_{4} & 3a_{3} & 2a_{2} & a_{1}\\[-2pt]
& 5a_{5} & 4a_{4} & 3a_{3} & 2a_{2}\\[-2pt]
&  & 5a_{5} & 4a_{4} & 3a_{3}%
\end{array}
\right\vert, \\
Q_{4}= &  \left\vert
\begin{array}
[c]{ccccccccc}%
a_{5} & a_{4} & a_{3} & a_{2} & a_{1} & a_{0} &  &  & \\[-2pt]
& a_{5} & a_{4} & a_{3} & a_{2} & a_{1} & a_{0} &  & \\[-2pt]
&  & a_{5} & a_{4} & a_{3} & a_{2} & a_{1} & a_{0} & \\[-2pt]
&  &  & a_{5} & a_{4} & a_{3} & a_{2} & a_{1} & a_{0}\\[-2pt]
& 10a_{5} & 6a_{4} & 3a_{3} & a_{2} &  &  &  & \\[-2pt]
&  & 10a_{5} & 6a_{4} & 3a_{3} & a_{2} &  &  & \\[-2pt]
&  &  & 10a_{5} & 6a_{4} & 3a_{3} & a_{2} &  & \\[-2pt]
&  &  &  & 10a_{5} & 6a_{4} & 3a_{3} & a_{2}\  & \\[-2pt]
&  &  &  & 5a_{5} & 4a_{4} & 3a_{3} & 2a_{2} & a_{1}%
\end{array}
\right\vert +\left\vert
\begin{array}
[c]{ccccccccc}%
a_{5} & a_{4} & a_{3} & a_{2} & a_{1} & a_{0} &  &  & \\[-2pt]
& a_{5} & a_{4} & a_{3} & a_{2} & a_{1} & a_{0} &  & \\[-2pt]
&  & a_{5} & a_{4} & a_{3} & a_{2} & a_{1} & a_{0} & \\[-2pt]
&  &  & a_{5} & a_{4} & a_{3} & a_{2} & a_{1} & a_{0}\\[-2pt]
& 10a_{5} & 6a_{4} & 3a_{3} & a_{2} &  &  &  & \\[-2pt]
&  & 10a_{5} & 6a_{4} & 3a_{3} & a_{2} &  &  & \\[-2pt]
&  &  & 10a_{5} & 6a_{4} & 3a_{3} & a_{2} &  & \\[-2pt]
&  &  & 5a_{5} & 4a_{4} & 3a_{3} & 2a_{2} & a_{1} & \\[-2pt]
&  &  &  &  & 10a_{5} & 6a_{4} & 3a_{3} & a_{2}%
\end{array}
\right\vert \\[5pt]
+ &  \left\vert
\begin{array}
[c]{ccccccccc}%
a_{5} & a_{4} & a_{3} & a_{2} & a_{1} & a_{0} &  &  & \\[-2pt]
& a_{5} & a_{4} & a_{3} & a_{2} & a_{1} & a_{0} &  & \\[-2pt]
&  & a_{5} & a_{4} & a_{3} & a_{2} & a_{1} & a_{0} & \\[-2pt]
&  &  & a_{5} & a_{4} & a_{3} & a_{2} & a_{1} & a_{0}\\[-2pt]
& 10a_{5} & 6a_{4} & 3a_{3} & a_{2} &  &  &  & \\[-2pt]
&  & 10a_{5} & 6a_{4} & 3a_{3} & a_{2} &  &  & \\[-2pt]
&  & 5a_{5} & 4a_{4} & 3a_{3} & 2a_{2} & a_{1} &  & \\[-2pt]
&  &  &  & 10a_{5} & 6a_{4} & 3a_{3} & a_{2} & \\[-2pt]
&  &  &  &  & 10a_{5} & 6a_{4} & 3a_{3} & a_{2}
\end{array}
\right\vert +\left\vert
\begin{array}
[c]{ccccccccc}%
a_{5} & a_{4} & a_{3} & a_{2} & a_{1} & a_{0} &  &  & \\[-2pt]
& a_{5} & a_{4} & a_{3} & a_{2} & a_{1} & a_{0} &  & \\[-2pt]
&  & a_{5} & a_{4} & a_{3} & a_{2} & a_{1} & a_{0} & \\[-2pt]
&  &  & a_{5} & a_{4} & a_{3} & a_{2} & a_{1} & a_{0}\\[-2pt]
& 10a_{5} & 6a_{4} & 3a_{3} & a_{2} &  &  &  & \\[-2pt]
& 5a_{5} & 4a_{4} & 3a_{3} & 2a_{2} & a_{1} &  &  & \\[-2pt]
&  &  & 10a_{5} & 6a_{4} & 3a_{3} & a_{2} &  & \\[-2pt]
&  &  &  & 10a_{5} & 6a_{4} & 3a_{3} & a_{2} & \\[-2pt]
&  &  &  &  & 10a_{5} & 6a_{4} & 3a_{3} & a_{2}%
\end{array}
\right\vert \\[5pt]
+ &  \left\vert
\begin{array}
[c]{ccccccccc}%
a_{5} & a_{4} & a_{3} & a_{2} & a_{1} & a_{0} &  &  & \\[-2pt]
& a_{5} & a_{4} & a_{3} & a_{2} & a_{1} & a_{0} &  & \\[-2pt]
&  & a_{5} & a_{4} & a_{3} & a_{2} & a_{1} & a_{0} & \\[-2pt]
&  &  & a_{5} & a_{4} & a_{3} & a_{2} & a_{1} & a_{0}\\[-2pt]
5a_{5} & 4a_{4} & 3a_{3} & 2a_{2} & a_{1} &  &  &  & \\[-2pt]
&  & 10a_{5} & 6a_{4} & 3a_{3} & a_{2} &  &  & \\[-2pt]
&  &  & 10a_{5} & 6a_{4} & 3a_{3} & a_{2} &  & \\[-2pt]
&  &  &  & 10a_{5} & 6a_{4} & 3a_{3} & a_{2} & \\[-2pt]
&  &  &  &  & 10a_{5} & 6a_{4} & 3a_{3} & a_{2}%
\end{array}
\right\vert.
\end{align*}

\item HY22's condition: ${R_1}=0\wedge R_{2}=0\wedge R_{3}\neq0$ where
$R_1=P_1$ and%
\begin{align*}
R_{2}= &  \frac{1}{a_{5}}\left\vert
\begin{array}
[c]{cccccccc}%
a_{5} & a_{4} & a_{3} & a_{2} & a_{1} & a_{0} &  & \\[-2pt]
& a_{5} & a_{4} & a_{3} & a_{2} & a_{1} & a_{0} & \\[-2pt]
&  & a_{5} & a_{4} & a_{3} & a_{2} & a_{1} & a_{0}\\[-2pt]
5a_{5} & 4a_{4} & 3a_{3} & 2a_{2} & a_{1} &  &  & \\[-2pt]
& 5a_{5} & 4a_{4} & 3a_{3} & 2a_{2} & a_{1} &  & \\[-2pt]
&  & 5a_{5} & 4a_{4} & 3a_{3} & 2a_{2} & a_{1} & \\[-2pt]
&  && 5a_{5} & 4a_{4} & 3a_{3} & 2a_{2} & a_{1} \\[-2pt]
&  &  &  & 20a_{5} & 12a_{4} & 6a_{3} & 2a_{2}%
\end{array}
\right\vert \\[5pt]
R_{3}= &  \frac{1}{a_{5}}\left\vert
\begin{array}
[c]{ccccccc}%
a_{5} & a_{4} & a_{3} & a_{2} & a_{1} & a_{0} & \\
& a_{5} & a_{4} & a_{3} & a_{2} & a_{1} & a_{0}\\
5a_{5} & 4a_{4} & 3a_{3} & 2a_{2} & a_{1} &  & \\
& 5a_{5} & 4a_{4} & 3a_{3} & 2a_{2} & a_{1} & \\
&  & 5a_{5} & 4a_{4} & 3a_{3} & 2a_{2} & a_{1}\\
&  & 20a_{5} & 12a_{4} & 6a_{3} & 2a_{2} & \\
&  &  & 20a_{5} & 12a_{4} & 6a_{3} & 2a_{2}%
\end{array}
\right\vert
\end{align*}
\end{enumerate}

From the above conditions, we make the following observations which
are also true in general.
\begin{enumerate}
\item YHZ's discriminant involves one  nested determinant;

\item HY21's discriminant involves a sum of several non-nested determinants;

\item HY22's discriminant involves one non-nested determinant.
\end{enumerate}

\subsection{Maximum degree of discriminants}

For the sake of simplicity, we use the following short-hands:

\begin{itemize}
\item $d_{\mathrm{YHZ}}$ : the maximum of the degrees of the polynomials
appearing in YHZ's conditions (\cite{1996_Yang_Hou_Zeng});

\item $d_{\mathrm{HY21}}$ : the maximum of the degrees of the polynomials
appearing in HY21's conditions (\cite{2021_Hong_Yang});

\item $d_{\mathrm{HY22}}$ : the maximum of the degrees of the polynomials
appearing in the new conditions (Theorem~\ref{thm:main_result}).
\end{itemize}

\begin{lemma}
\label{lem:compare} Let $d_{\mathrm{YHZ}}(\boldsymbol{\mu})$,$d_{\mathrm{HY21}%
}(\boldsymbol{\mu})$ and $d_{\mathrm{HY22}}(\boldsymbol{\mu})$ denote the
maximum degrees of the polynomials appearing in YHZ's condition, HY21's
condition and HY22's condition for a given $\boldsymbol{\mu}=(\mu_{1}%
,\ldots,\mu_{m})\in\mathcal{M}(n)$, respectively. Then we have:

\begin{enumerate}
\item Under some minor and reasonable assumption (see \cite[Assumption
2]{2021_Hong_Yang_arxiv}),
\[
d_{\mathrm{YHZ}}(\boldsymbol{\mu})= \left\{
\begin{array}
[c]{ll}%
\prod\limits_{j=1}^{\mu_{2}}(2\,m_{j}-1)\left\{
\begin{array}
[c]{lll}%
1 & \text{if} & \mu_{1}=\mu_{2}\\
1+\frac{2}{2m_{\mu_{2}}-1} & \text{if} & \mu_{1}=\mu_{2}+1\\
\left(  2\left(  \mu_{1}-\mu_{2}\right)  -1\right)  & \text{if} & \mu_{1}%
>\mu_{2}+1
\end{array}
\right. \;\;\;\;\;\;\geq\;\;2n+3^{\mu_{2}}-4\mu_{2}, & \text{for}\ m>1;\\
2n-1, & \text{for}\ m=1,
\end{array}
\right.
\]

where $m_{i}=\#\{\mu_{k}:\,\mu_{k}\ge i\}$;

\item $d_{\mathrm{HY21}}(\boldsymbol{\mu})=2n-1;$

\item $d_{\mathrm{HY22}}(\boldsymbol{\mu})=2n-2.$
\end{enumerate}
\end{lemma}

\begin{proof}
\

\begin{enumerate}
\item When $m=1$, $\boldsymbol{\mu}=(n)$. In this case, the condition for the
polynomial having multiplicity structure $\boldsymbol{\mu}$ is given by the
$0$-th,\ldots,$(n-1)$-th subdiscriminants. Thus the maximum degree
$d_{\mathrm{YHZ}}(\boldsymbol{\mu})$ is $2n-1$, achieved at the $0$-th subdiscriminant.

When $m>1$, see \cite[Appendix]{2021_Hong_Yang_arxiv} for a detailed proof.

\item Recall that HY21's condition consists of two parts: (i) the
$0$-th,\ldots,$(n-m)$-th subdiscriminants whose highest degree is $2n-1$; (ii)
the multiplicity discriminant given by
\[
\sum_{{\sigma}\in S_{p}}\operatorname*{dp}\left[
\begin{array}
[c]{c}%
x^{n-\mu_{m}-1}F\\
\vdots\\
x^{0}F\\
x^{n-1}F^{\left(  \sigma_{_{1}}\right)  }/\sigma_{1}!\\
\vdots\\
x^{0}F^{\left(  \sigma_{_{n}}\right)  }/\sigma_{n}!
\end{array}
\right]
\]
where $p=(\underset{\mu_{1}}{\underbrace{\mu_{1},\ldots,\mu_{1}}},\ldots,$
$\underset{\mu_{m}}{\underbrace{\mu_{m},\ldots,\mu_{m}}})$ and $S_{p}$ is the
set of all permutations of $p$. It is easy to see that the degree of the
multiplicity discriminant is $2n-\mu_{m}$. Hence the maximum degree of the
above discriminants is $2n-1$.

\item HY22's condition only consists of the multiplicity discriminants given
by
\[
D\left(  \boldsymbol{\gamma}\right)  =\frac{1}{a_{n}}\operatorname*{dp}\left[
\begin{array}
[c]{l}%
F^{(0)}x^{\gamma_{0}-1}\\
~~~~\ \vdots\\
F^{(0)}x^{0}\\\hline
F^{(1)}x^{\gamma_{1}-1}\\
~~~~\ \vdots\\
F^{(1)}x^{0}\\\hline
~~~~\ \vdots\\\hline
F^{(s)}x^{\gamma_{s}-1}\\
~~~~\ \vdots\\
F^{(s)}x^{0}%
\end{array}
\right]
\]
where $\boldsymbol{\gamma}=(\gamma_{1},\ldots,\gamma_{s})$ ranges over
$\overline{\boldsymbol{\mu}}\prec_{\operatorname*{lex}}\cdots\prec
_{\operatorname*{lex}}(n)$. Note that the highest degree is achieved when
$\boldsymbol{\gamma}=(n)$. In this case, the degree of $D\left(
\boldsymbol{\gamma}\right)  $ is $2n-2$.
\end{enumerate}
\end{proof}

\begin{remark}
It is noted that in HY21's condition, the multiplicity discriminant is always
divisible by the leading coefficient $a_{n}$ and thus with this division
carried out, the degree can be made smaller by $1$.
\end{remark}

By Lemma \ref{lem:compare}, the maximum degree in YHZ's condition grows
exponentially with respect to $n$ while the maximum degrees in HY21 and HY22's
conditions grow linearly. Below we show a comparison with examples where
$n<10$.

\begin{minipage}{1.5\textwidth}
\begin{minipage}[!b]{0.3\textwidth}
\makeatletter\def\@captype{table}
\begin{center}
\begin{tabular}{cccc}\toprule
$n$ & $d_{\mathrm{YHZ}}$ & $d_{\mathrm{HY21}}$ & $d_{\mathrm{HY22}}$ \\
\hline
3 & 5 & 5 & 4 \\[-2pt]
4 & 9 & 7 & 6 \\[-2pt]
5 & 15 & 9 & 8 \\[-2pt]
6 & 27 & 11 & 10 \\ [-2pt]
7 & 45 & 13 & 12 \\ [-2pt]
8 & 81 & 15 & 14 \\ [-2pt]
9 & 135 & 17 & 16 \\\bottomrule
\end{tabular}
\caption{Comparison on the maximal degrees\\of polynomials in the conditions generated with\\
the three methods}
\end{center}
\end{minipage}
\begin{minipage}[!b]{0.3\textwidth}
\makeatletter\def\@captype{figure}
\begin{center}
\includegraphics[width=0.9\textwidth]{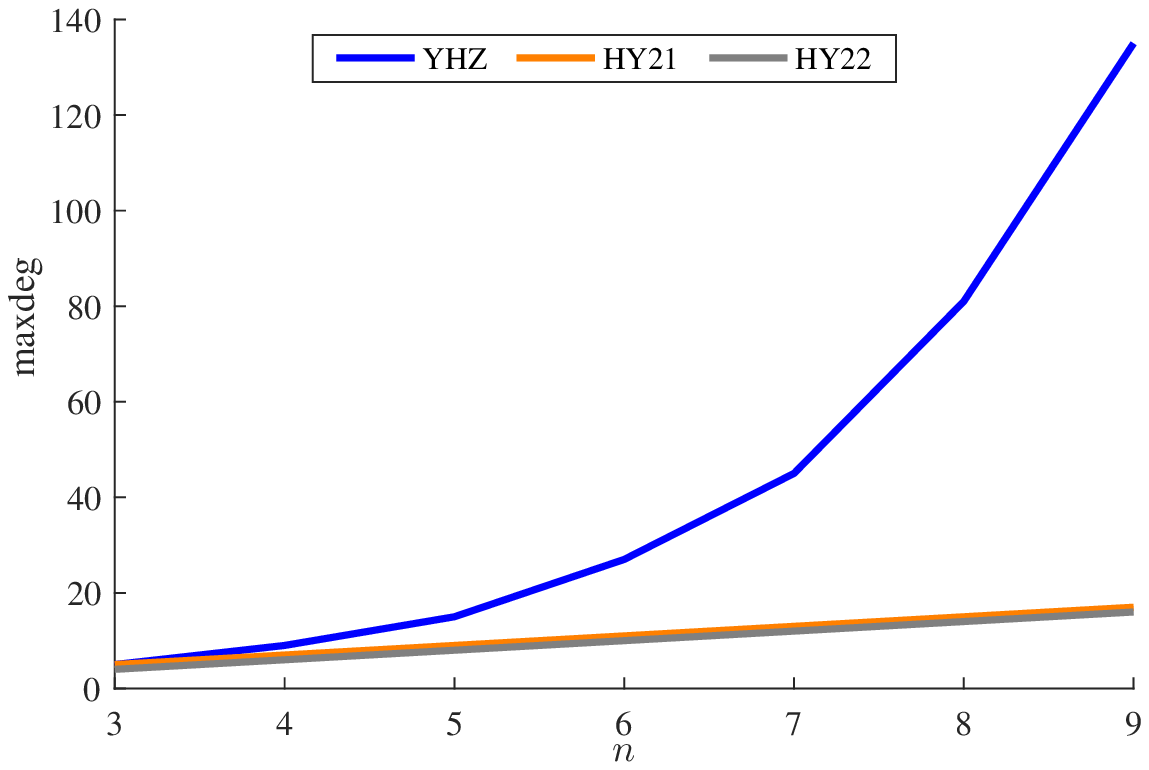}
\caption{An illustration on the changes of maximal degrees of polynomials in the conditions generated with the three methods along with the degree $n$}
\end{center}
\end{minipage}
\end{minipage}

\medskip\noindent\Acknowledgements{Hoon Hong's work was
supported by National Science Foundations of USA (Grant Nos: 2212461 and 1813340).
 Jing Yang's work was
supported by National Natural Science Foundation of China (Grant Nos.:
12261010 and 11801101).}

\end{document}